\documentclass[a4paper,11pt]{article}

\usepackage{graphicx,amssymb,amsmath,amsthm}

\usepackage{fullpage}
\usepackage[english]{babel}
\usepackage{subfigure}
\usepackage{color}
\usepackage{enumerate}
\usepackage{hyperref}

\renewcommand{\emptyset}{\varnothing}

\def\range#1{\in \left[ #1 \right]  }
\def\R{{\mathbb R}}
\def\N{{\mathbb N}}

\newtheorem{theorem}{Theorem}
\newtheorem{lemma}[theorem]{Lemma}

\theoremstyle{definition}

\title{Colored Ray Configurations}

\author{
Ruy Fabila-Monroy\thanks{Dept. de Matem\'{a}ticas. Centro de Investigaci\'{o}n y de Estudios Avanzados del Instituto Polit\'{e}cnico Nacional, M\'{e}xico D.F., M\'{e}xico, {\tt ruyfabila@math.cinvestav.edu.mx}} \and
Alfredo Garc\'{\i}a\thanks{Dept. de M\'{e}todos Estad\'{\i}sticos, IUMA, Universidad de Zaragoza, Spain, {\tt \{olaverri, jtejel\}@unizar.es } }  \and
Ferran Hurtado\thanks{Dept. de Matem\`{a}tica Aplicada II, Universitat Polit\`{e}cnica de Catalunya, Barcelona, Spain.} \and
Rafel Jaume\thanks{Palma, Spain, {\tt rafel.jd@gmail.com}} \and 
Pablo P\'{e}rez-Lantero\thanks{Departamento de Matem\'atica y Ciencia de la Computaci\'on, Universidad de Santiago, Santiago, Chile, {\tt pablo.perez.l@usach.cl}. Partially supported by projects CONICYT FONDECYT/Regular 1160543 (Chile), 
and Millennium Nucleus Information and Coordination in Networks ICM/FIC RC130003 (Chile).} \and
Maria Saumell\thanks{Department of Mathematics and European Centre of Excellence NTIS (New Technologies for the Information Society), University of West Bohemia, Pilsen, Czech Republic, {\tt saumell@kma.zcu.cz}} \and
Rodrigo I. Silveira\thanks{Dept.\ Matem\`atica Aplicada II, Universitat Polit\`ecnica de Catalunya, \texttt{rodrigo.silveira@upc.edu}} \and
Javier Tejel\footnotemark[2]  \and
Jorge Urrutia\thanks{Instituto de Matem\'{a}ticas, Universidad Nacional Aut\'{o}noma de M\'{e}xico, M\'{e}xico D.F., M\'{e}xico, {\tt urrutia@matem.unam.mx}}
}

\begin{document}
\thispagestyle{empty}
\maketitle

%\linenumbers

\begin{abstract}
We study the cyclic color sequences induced at infinity by colored rays with apices being a given balanced finite
bichromatic point set. 
We first study the case in which the rays are required to be pairwise disjoint. 
We derive a lower bound on the number of color sequences that can be realized from any such fixed point set and examine color sequences that can be realized regardless of the point set, exhibiting negative examples as well.
We also provide a tight upper bound on the number of configurations that can be realized from a point set, and point sets for which there are asymptotically less configurations than that number.  
In addition, we provide algorithms to decide whether a color sequence is realizable from a given point set in a line or in general position. 
We address afterwards the variant of the problem where the rays are allowed to intersect. 
We prove that for some configurations and point sets, the number of ray crossings must be $\Theta(n^2)$ and study then configurations that can be realized by rays that pairwise cross. 
We show that there are point sets for which the number of configurations that can be realized by pairwise-crossing rays is asymptotically smaller than the number of configurations realizable by pairwise-disjoint rays. 
We provide also point sets from which any configuration can be realized by pairwise-crossing rays and show that there is no configuration that can be realized by pairwise-crossing rays from every point set. 
\end{abstract}

\section*{In memoriam: Ferran Hurtado (1951-2014)}

This paper is dedicated to the memory of our dear friend and colleague Ferran Hurtado that left us unexpectedly in 2014.

\section{Introduction} 
\label{sec:intro}

Let $S=\{p_1, \ldots , p_n\}$ be a set of points in the plane in general position. A set $H=\{h_1, \ldots , h_n\}$ such that the apex of $h_i$ is $p_i$, is called a set of \emph{rays from $S$}, $i=1, \ldots , n$. The elements of $H$ induce at infinity a cyclic permutation defined by the indices of the rays. 
In this paper we will also deal with colored point sets, in this case, a ray with apex $p_i$, inherits the color assigned to $p_i$.
Requiring that the rays in $H$ are pairwise-disjoint, how many different permutations can always be obtained disregarding the geometry of $S$? 
Is there any upper bound for their number for all sets of $n$ points? What happens in some particular configurations, for example when $S$ is in convex position? These problems---and several related questions---were introduced by Hurtado et al.~\cite{GHTU07}.

A clear motivation for the research in~\cite{GHTU07} was the extensive investigation on counting non-crossing geometric graphs of several families, such as spanning cycles, perfect matchings, triangulations and many more, and on estimating how large these numbers can get \cite{AHHKV07,DSST13,HSSTW11,SS09,SSW13,SW06}.
On the other hand, arrangements of rays have appeared in graph representation: \emph{Ray Intersection Graphs} are those in which there is a node for every ray in a given set, two of which are adjacent if they intersect~\cite{CCL13,FMM13,MTU10}.
Finally, on the applied side, it is worth mentioning recent work on sensor networks in the plane in which each sensor coverage region is an arbitrary ray \cite{KYZ13}. The rays act as barriers for detecting the movement between regions in the arrangement.

The work in~\cite{GHTU07} studies, among other variants, the number $\sigma(S)$ of different cyclic permutations of $\{1,2,\ldots,n\}$ that can be induced by sets of non-crossing rays from a set $S$ of $n$ labeled points. 
They show that \[  \sigma^{\min}(n)= \Omega^* (2^n) \cap O^* (3.516^n) \text{ and that } \sigma^{\max}(S) = \Theta ^* (4^n ),\footnote{Throughout the paper, the $O^*()$ notation omits subexponential factors.}\]
where $\sigma^{\min}$ and $\sigma^{\max}$ are the minimum and the maximum, respectively, of $\sigma(S)$ taken over all labeled sets $S$ of $n$ points in the plane. 

In this paper we consider a natural variation on the problem introduced in~\cite{GHTU07}.
The point set consists now of red and blue points, and the ray we shoot from a point inherits its color. 
The rays are first required to be pairwise disjoint. 
We investigate the bichromatic circular sequences that the colored rays induce at infinity: 
We study how many different color patterns can always be obtained and  how many color alternations, depending
on  the generality of the position of the points (Section~\ref{subsec:lower-bound}); we also investigate whether there are
color patterns that cannot be realized for some particular point set (\emph{forbidden}), or that can be attained by every point set (\emph{universal}) (see Section~\ref{subsec:realizations}). 
We provide  as well decision algorithms for some particular cases (Section~\ref{subsec:deciding}).
Section~\ref{sec:nondisjoint} is concerned with sets of rays that are not necessarily disjoint, for which 
feasibility questions are not interesting. 
We describe point sets of size $n$ from which any set of rays realizing a certain configuration must produce $\Theta(n^2)$ crossings. 
We study then a variant of the problem where we require each pair of rays to cross. 
We prove that the number of configurations realizable in this scenario can be asymptotically smaller than the number of configurations realizable by pairwise-disjoint rays. 
We  also show that for this variant there exist universal point sets, yet there is no universal configuration.

\section{Notation and definitions}

Henceforth, $\N$ will denote the positive integers. 
Given $k \in \N$, we denote by $[k]$ the set of integers~$\{1,\ldots, k\}$. 
Let $S=R\cup B$ be a finite bichromatic point set,
where $R$ is the set of \emph{red} points, and $B$ the set of \emph{blue} points of $S$.
We require $S$ to be {\em balanced} ($|R|=|B|$), which is the variant that has received most attention in the family of problems on \emph{red-blue point sets}~\cite{KK03}.

Given a set $H$ of rays from $S$, let $C(S,H)$ denote the circular sequence of length $|S|$ in the alphabet $\{r, b\}$ induced by the rays at infinity, taken in clockwise order.
Equivalently, we can take any circle large enough to enclose~$S$, and think of $C(S,H)$ as the sequence of colors of the intersection points of the rays with the circle, in clockwise order along the boundary.

Given $n\in\N$, a {\em configuration} is a circular sequence of $2n$ elements in the alphabet $\{r, b\}$ consisting of $n$ red elements and $n$ blue elements. 
Each element of a configuration is also called a \emph{position}. 
We assume hereafter that any configuration $C$ starts with a red element and ends with a blue one. 
Notice that $C$ can be partitioned into $2k$ monochromatic \emph{blocks}, each with $r_1, b_1, \ldots , r_k, b_k$ elements respectively.
 
We say that $k$ is the {\em alternation number} of~$C$.
Hence, $C$ can be identified with the tuple $(r_1,b_1,r_2,b_2,\ldots,r_k,b_k)$,
where $r_i$ and $b_i$ are the number of elements in the red and blue blocks respectively, for $i\range{k}$.

Let $\Gamma(n)$ denote the number of configurations;
this is equivalent to the number of \emph{binary balanced necklaces} of length $2n$. 
A binary balanced necklace is an equivalence class of $2n$-character strings on the alphabet $\{r,b\}$ with the same number of occurrences of $r$ and $b$, where two strings are equivalent if one can be obtained from the other by a cyclic rotation.
These objects where studied in more generality already in~\cite{M1872}.
Binary necklaces are counted in~\cite[Example~37.4]{vanLint} as an application of P\'{o}lya enumeration theorem or Burnside's lemma. 
This formulation, for the balanced case, yields
\[ \frac{1}{2n} {2n\choose n} \le \Gamma(n)= \frac{1}{2n}   \sum_{d | n}  \varphi(d) {2n/d \choose n/d} \le {2n\choose n}, \] where $\varphi$ is Euler's totient function.
Consequently, we have $\Gamma(n) = \Theta ^* (4^n)$.

Given $S$ and a configuration $C$, we say that
$C$ is {\em feasible} from $S$ if there exists a set $H$ of pairwise-disjoint rays from $S$
such that $C=C(S,H)$.  
We also say in this situation, and when three rays in $H$ may intersect in one point, that $C$ is \emph{realized} by $H$ (from $S$). 
See Figure~\ref{fig:ilustracion1} for an example.
We say that a configuration is \emph{universal} if it is feasible from any point set of the corresponding size. 
We say that a configuration is \emph{forbidden} for a point set if it cannot be realized from it.

\begin{figure}[!htbp]
\centering
\includegraphics[page=12]{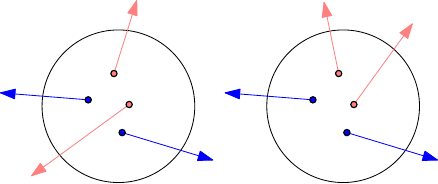}
\caption{A point set and a realization of the configurations $rbrb$ (left) and $rrbb$ (right).
} \label{fig:ilustracion1}
\end{figure}

Given a directed line (or a ray) $\ell$, let $\ell^+$ and $\ell^-$ denote the sets of
points to the right and to the left of $\ell$, respectively. 
Given a point $p$ and a vector $v$ of the plane,
let $h(p,v)$ denote the ray $\{p+t\cdot v~|~t \in \R, t \ge 0\}$ with apex~$p$.
Let~$H'$ be a set of rays such that for every pair $h_1,h_2\in H'$
the intersection $h_1\cap h_2$ is either empty, one of the apices, or contains an infinite number of points.
In this case we say that~$H'$ is a set of {\em non-crossing} rays.  

We say that a point set $S$ is in \emph{strong general position}, if it is in general position, and no different pairs of points define parallel lines.

\section{Disjoint rays}

In this section we study configurations that can be realized by sets of pairwise-disjoint rays.
First, we give lower bounds on the number of color patterns and color alternations that can always be obtained, depending
on the generality of the position of the points.
Second, we investigate whether there are color patterns that are forbidden or universal, and exhibit several positive and negative examples. 
Finally, we provide algorithms to decide whether a configuration is feasible from a given point set in a line or in general position. 

Unless stated otherwise, $S$ is a balanced bichromatic point set of total size $2n$.
Let $\gamma(S)$ denote the number of different feasible configurations $C(S,H)$
over all the sets $H$ of rays from~$S$. 
Let $\gamma^{\min}_{\rm{sgp}}(n)$ and $\gamma^{\max}_{\rm{sgp}}(n)$ be the minimum
and the maximum of $\gamma(S)$, respectively, taken over all balanced bichromatic
sets~$S$ of $2n$ points in the plane in strong general position. 
The notations $\gamma^{\min}_{\rm{gp}}(n)$ and $\gamma^{\max}_{\rm{gp}}(n)$ correspond \emph{mutatis mutandis} to the case in which only general position is required.

\subsection{Bounds on $\gamma(S)$ and on the alternation number} 
\label{subsec:lower-bound}

In this subsection, we provide lower bounds on $\gamma^{\min}_{\rm{sgp}}(n), \gamma^{\min}_{\rm{gp}}(n)$ and prove that $\gamma^{\max}_{\rm{sgp}}(n)=\Gamma(n)$. 
In addition, we give a tight lower bound on the maximum alternation number that can be attained from any point set in strong general position, and an upper bound for $\gamma^{\min}_{\rm{gp}}(n)$.

We first prove a lower bound on the number of feasible configurations, and a tight lower bound for the number of alternations attainable from any point set in strong general position.

\begin{theorem}\label{theo:lower-bound-S}
For every bichromatic point set $S=R\cup B$ in strong general position, it holds that $\gamma(S)=\Omega(2^{\sqrt{n}}/n)$. 
Hence, $\gamma^{\min}_{\rm{sgp}}(n)=\Omega(2^{\sqrt{n}}/n)$.
\end{theorem}

\begin{proof}
By the Ham-Sandwich Cut Theorem~\cite{GOR04}, there exists a (directed) line $\ell$
such that $|R^+|=|B^-|=\lfloor n/2\rfloor$, where $R^+=R\cap \ell^+$ and $B^-=B\cap \ell^-$.
Let $m=\lfloor n/2\rfloor$. We can assume, via a virtual rotation of the coordinate system, that $\ell$ is the positively oriented $x$-axis.
Since $|R^+|=|B^-|=m$, there exists a non-crossing geometric
perfect matching on $R^+\cup B^-$, that is,
$m$ pairwise-disjoint straight-line segments $e_1,e_2,\ldots,e_m$
such that~$e_i$ connects an element of $R^+$ with an element of $B^-$ and also intersects
$\ell$, for $i \range{m}$.

Assume without loss of generality that the points $e_1\cap\ell,e_2\cap\ell,\ldots,e_m\cap\ell$ are sorted from left to right.
Using the Erd\H{o}s-Szekeres Theorem on sequences~\cite{ES35}, there exist $k=\Omega(\sqrt{m})=\Omega(\sqrt{n})$
indices $1\le i_1< i_2<\ldots<i_k\le m$
such that the clockwise angles from the segments $e_{i_1},e_{i_2},\ldots,e_{i_k}$ to $\ell$ 
are either monotonically increasing or monotonically decreasing.
Assume without loss of generality that the angles are monotonically decreasing and observe that, because of the assumption of strong general position, they decrease \emph{strictly}.
Let $p_j\in B^-$ and $q_j\in R^+$ denote the endpoints of $e_{i_j}$, for $j\range{k}$.
Let $H_p=\{h(p_j,p_j-q_j)~|~j\range{k} \}$ and
$H_q=\{h(q_j,p_j-q_j)~|~j\range{k}\}$, and observe that the elements
of $H_p$ (resp.\ $H_q$) are pairwise disjoint.
Let $H_0$ be a set of rays from $S \setminus (\{p_j \mid j \range{k} \} \cup \{q_j \mid j \range{k} \})$ such that
every element of $H_0$ does not intersect, and is not parallel to, any element of $H_p\cup H_q$;
it is clear that such a set of rays $H_0$ always exists,
and that $H_p\cup H_q\cup H_0$ is a set of non-crossing rays. 
Furthermore,
we can perturb the elements of $H_p\cup H_q$ in $2^k$ different forms to obtain a set $H$ of pairwise-disjoint rays from $S$.
The perturbation is as follows: For a small enough angle $\varepsilon>0$ and $j\range{k}$, rotate both $h(p_j,p_j-q_j)$ and $h(q_j,p_j-q_j)$ with angle $\varepsilon$ around their apices, either clockwise or counterclockwise. 
Then, among all sets $H$,
the configuration $C(S,H)$ is different for at least $2^k /2n=\Omega(2^{\sqrt{n}}/n)$
of them. 
The claim follows.
\end{proof}

We now look at the alternation number.

\begin{theorem}\label{theo:sqrt-alternation}
For every bichromatic point set $S=R\cup B$ in strong general position, 
there exists a set $H$ of pairwise-disjoint rays from $S$ such that
the alternation number of $C(S,H)$ is $\Omega(\sqrt{n})$. 
This bound is tight.
\end{theorem}

\begin{proof}
Observe that the sets of rays from $S$ generated in the proof
of Theorem~\ref{theo:lower-bound-S} yield $\Omega(\sqrt{n})$ color alternations.
To prove that this bound is tight, let $n=k^2$ for some $k \in \N$, and
$R$ and $B$ be defined as follows. 
For $i\range{k}$, let $B_i=\{(2(i-1)+j/n^2,0) \mid j\range{k} \}$, $B=\bigcup_{i\range{k}} B_i$, and $R=\{(j/n,1)\mid j\range{n} \}$.
Let~$s_i$ be the shortest segment covering $B_i$, for $i\range{k}$, and
$s'$ the shortest segment covering $R$ (thus $|s_i| \approx n^{-3/2}$ and $|s'| \approx 1$).
Observe that no two pairwise disjoint rays from elements of $R$ can intersect the
same segment~$s_i$. 
Furthermore, no two pairwise disjoint rays $b_1 \in B_{i}$ and $b_2 \in B_{j}$ with $i,j \in [k]$, $i \neq j$, can intersect $s'$. 
Therefore, any set $H$ of pairwise-disjoint rays
from $S=R\cup B$ is such that $C(S,H)$ has $O(k)=O(\sqrt{n})$ alternations.
Finally, observe that some infinitesimal perturbation of the points moves them to strong general position, and still yields
the same upper bound construction.
\end{proof}

Without the assumption of strong general position many of the segments in the matching used in the proof of Theorem~\ref{theo:lower-bound-S}, or even all of them, might be parallel, which disables the construction in that proof. 
It is easy to see that given a set of $n$ red points above the $x$-axis and a set of $n$ blue points below the $x$-axis, whose union is in general position, one can always obtain a bichromatic matching of size at least $\sqrt{n}$, such that the angles defined by the matched segments and the $x$-axis are different. 
This combines with the technique of Theorem~\ref{theo:lower-bound-S} to yield an
$\Omega(2^{n^{1/4}}/n)$
lower bound for the number of different configurations realizable from point sets in general position. We can do better with a related, yet different, approach.

\begin{theorem}\label{theo:lower-bound-S-gp}
For every bichromatic point set $S=R\cup B$ in general position, ${\gamma(S)=\Omega(2^{n^{1/3}}/n)}$. 
Hence, $\gamma^{\min}_{\rm{gp}}(n)=\Omega(2^{n^{1/3}}/n)$.
\end{theorem}

\begin{proof}
We start as in the proof of Theorem~\ref{theo:lower-bound-S} and obtain a bichromatic non-crossing geometric
perfect matching of a set  $R^+$ of $m$ red points below the $x$-axis, and a set $ B^-$
of $m$ blue points above the $x$-axis, with $m=\Theta(n)$.
Now, using a generalized version of the Erd\H{o}s-Szekeres Theorem on sequences\footnote{Let $n>s \cdot r \cdot p$. Any sequence of $n$ numbers contains a strictly increasing subsequence with at least $s+1$ elements, a strictly decreasing subsequence with at least $r+1$ elements, or a constant subsequence of length greater than $p$. }~\cite{Juk11}, there exist $k=\Omega(m^{1/3})=\Omega(n^{1/3})$
indices $1\le i_1< i_2<\ldots<i_k\le m$
such that the clockwise angles from the segments $e_{i_1},e_{i_2},\ldots,e_{i_k}$ to the $x$-axis
are either monotonically \emph{strictly} increasing, or monotonically \emph{strictly} decreasing, or all equal.
Let $S_e$ denote the set of endpoints of $e_{i_j}$ for $j\range{k}$.

In the first two cases we apply the technique in the proof of Theorem~\ref{theo:lower-bound-S}. 
It remains to consider the case in which $e_{i_1},e_{i_2},\ldots,e_{i_k}$ are all parallel.
Let us start with a line $\ell_p$ through each $p\in S$, in the direction of the segments, and then rotate~$\ell_p$ around~$p$ an infinitesimal angle $\varepsilon$ for all $p \in S$, in such a way that none of them contains two points. 
Rotating the whole construction if necessary, assume that the new lines $\ell'_p$ are vertical. 
Observe that the lines corresponding to the endpoints of a segment $e_{i_j}$ are now different and consecutive in the horizontal order, for $j \in [k]$. 
Now, shoot vertically and downwards a ray from every point in $S\setminus S_e$. 
For each $e_{i_j}$ with $j \in [k]$, we can independently decide for its endpoints whether we shoot a red ray upwards and a blue ray downwards, or reversely. 
This yields $\Omega(2^{n^{1/3}}/n)$ different configurations.
\end{proof}

We continue by showing the existence of point sets from which every configuration is possible. 
We say that such a point set is \emph{universal (for pairwise-disjoint rays)}.

\begin{theorem}\label{theo:universal-point-set}
	For every $n \in \N$, there exists a bichromatic point set $S=R\cup B$ in strong general position such that every configuration is feasible. 
	Hence, $\gamma^{\max}_{\rm{sgp}}(n)=\Gamma(n)$.
\end{theorem}

\begin{proof}
	We set $R=\bigl\{(1,1),(2,1),\ldots,(n,1)\bigr\}$ and $B=\bigl\{(1/n,0),(2/n,0),$ $\ldots,(n/n,0)\bigr\}$, and label the points $(1,i)=p_i$ and $(j/n,0)=q_j$ for all $i,j \in [n]$.
	Let $C$ be any configuration $(r_1,b_1,r_2,b_2,\ldots,r_k,b_k)$ with $k \in \N$.
	We show that we can draw a set $H_B$ of rays from $B$ such that the elements of $H_B$ are grouped into $k$ groups, the $i$th group consists of $b_i$ parallel rays, and the groups split $R$ into $k$ blocks such that the $j$th block (from left to right) consists of $r_j$ points.
	Namely, let $H_B=\bigl\{h\bigl(q_j,(t_j,1)\bigr)~|~j\range{n}\bigr\}$,
	where $t_j=r_1$ if $j \in [b_1]$, and $t_j=\sum_{k=1}^{s} r_k$ if $j\in [n]\setminus [b_1]$, and $s$ is the largest index such that $\sum_{k=1}^{s-1} b_k < j$.
	
	We give a precise and detailed construction because it will be used again in the proof of Theorem~\ref{fcuniv}.
	Let~$H_B'$ be an infinitesimal perturbation of $H_B$ such that no pair of rays are parallel. 
	Let $a_j$ be the intersection point of the lines supporting the rays from $H_B'$ shot from $q_{j-1}$ and $q_{j}$ for all $j\in \{2,\ldots,n\}$, and let $a_1$ be the intersection of the $y$-axis with the line supporting the ray from $H_B'$ shot from $q_1$ (see Figure~\ref{fig:crosuniv}).
	
	Let $H_R$ be the set consisting, for each $j\in \{2,\ldots,n\}$, of the rays $h(p_i,p_i-a_j)$ for all $i \in [n]$ such that $p_i$ is contained in the wedge defined by lines supporting the rays shot from $q_{j-1}$ and $q_{j}$ and containing a ray with direction $(t_j+t_{j-1})/2$, and the rays $h(p_i,p_i-a_1)$ for all $i \in [n]$ such that $p_i$ is contained in the wedge defined by $OY$ and the line supporting ray shot from $q_1$, and containing a ray with direction $((0,1)+(p_i-a_1))/2$. 
	
	\begin{figure}[!htb]
		\centering
		\includegraphics[page=2]{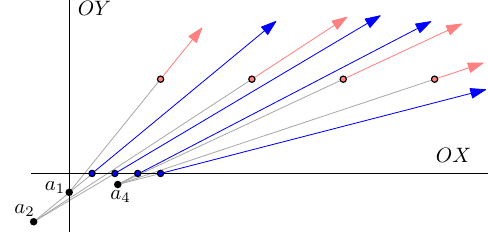}
		\caption{A universal point set for full-crossing rays. Note that $a_3$ falls out of the picture.}
		\label{fig:crosuniv}
	\end{figure}
	
	Clearly, $H=H_B' \cup H_R$ is a set of pairwise-disjoint rays that realize $C$. 
	The constructed point set can be perturbed to lie in strong general position in a way that the proof carries out. 
\end{proof}

In contrast with the previous theorem, the lower bound for the number of feasible configurations is asymptotically smaller than $\Gamma(n)$.

\begin{theorem}
 There are point sets for which the number of feasible configurations is asymptotically smaller than $\Gamma(n)$. 
 More precisely, $\gamma^{min}_{gp}(n)=O(2^{d\sqrt{n}\log(n)})$, for some constant $d>0$.
\end{theorem}

\begin{proof}	
	Let $S$ be the point set described in the proof of Theorem~\ref{theo:sqrt-alternation}.
	We now show that there are $O(2^{d\sqrt{n}\log(n)})$ feasible configurations from $S$, for some constant $d>0$.
	
	It follows easily from the proof of Theorem~\ref{theo:sqrt-alternation} that any configuration feasible from $S$ has at most $6\sqrt{n}+c$ alternations, for some constant $c$. 
	Let us count the number of linear sequences with~$i$ alternations for $i \in [6\sqrt{n}+c]$ representing a configuration, which is an upper bound for the number of configurations feasible from $S$. 
	We assume without loss of generality that every sequence starts with a red block (and ends with a blue block) and choose then the $i-1$ positions where the remaining changes of colors are produced. 
	This can be done in at most ${ 2n \choose i-1  }$ 	 ways.
	If $n$ is sufficiently large, we have that  \[{ 2n \choose i-1  } \le { 2n \choose {6\sqrt{n}+c}  } \le \left( \frac{2ne}{ 6\sqrt{n}+c}\right)^{6\sqrt{n}+c} \le \left( \frac{e\sqrt{n}}{3}  \right)^{6\sqrt{n}+c} \le  n^{3\sqrt{n}+c},\]
	for all $i \in [6\sqrt{n}+c]$, where we use the bound on binomial coefficients $ {n \choose k} \le \left((e n) /{k}\right)^k$.
	Hence, the total number of sequences is upper-bounded by
	\[ \sum \limits_{i=2}^{6\sqrt{n}+c} {2n \choose i-1}  \le (6\sqrt{n}+c)n^{3\sqrt{n}+c}, \]
	and the claim follows. 
\end{proof}

\subsection{Realizing configurations}\label{subsec:realizations}

We study in this section universal and non-universal configurations.
Observe that, as a consequence of Theorem~\ref{theo:sqrt-alternation}, configurations with $\omega(\sqrt{n})$ alternations are not realizable from every point set.
Further note that given any point set $S=R\cup B$ with $|R|=|B|=n$,
the configuration $(n,n)$ is always realizable: Rotate $S$ so that no pair of points lie in a horizontal line and draw from each red point a ray
oriented to the right, and from each blue point a ray oriented to the left. 
The resulting rays are pairwise-disjoint and satisfy $C(S,H)=(n,n)$.

We first investigate configurations with alternation number $2$.
If follows easily from the Ham Sandwich Theorem that the configuration $(\lfloor n/2\rfloor,\lfloor n/2\rfloor,\lceil n/2\rceil,\lceil n/2\rceil)$ is universal for every $n\in \N$.
Note that this configuration is as balanced as possible.
We now show that any point set in general position can yield some totally-unbalanced configurations as well. 

\begin{theorem}\label{theo:universal-conf2}
For every bichromatic point set ${S=R\cup B}$ in general position and any $t\in [n-1]$,
either the configuration $(n-1,n-t,1,t)$ or the configuration $(n-t,1,t,n-1)$ is
feasible.
\end{theorem}

\begin{proof}
Let $p\in S$ be a vertex of the convex hull of $S$. 
We show that if $p\in R$, then $(n-1,n-t,1,t)$ is feasible. 
If $p\in B$, it can be shown analogously that $(n-t,1,t,n-1)$ is feasible. 
Assume then that $p\in R$, and let $q\in B$ be a point such that $|B\cap \ell^-|=t-1$, where $\ell$ is the line
passing through $p$ and $q$.
Let 
\begin{align*}
H'= \; & \{h(p,p-q)\} \\ 
\cup \;& \{h(u,q-p)~|~u\in R\setminus\{p\}\} \\ 
\cup \;& \{h(u,p-q)~|~u\in B\},
\end{align*}
which is a set of non-crossing rays. 
After rotating each ray of $H$ a small angle $\varepsilon \ne 0$ so that $h(q,p-q)\setminus\{q\} \subset \ell^-$, it holds that $C(S,H)=(n-1,n-t,1,t).$
\end{proof}

In contrast with the previous results, some configurations with alternation number $2$ are not universal: 

\begin{theorem} \label{prop:unbalanced}
For any $n \in \N,\, n \ge 10$, there exist bichromatic point sets $S=R\cup B$ such that no configuration $(r_1,b_1,r_2,b_2)$ with either $$n-2>\max \{r_1,r_2\}>\max \{b_1,b_2\}+1 \text{,  or } $$ $$n-2>\max \{b_1,b_2\}>\max \{r_1,r_2\}+1$$ is feasible. 
\end{theorem}

\begin{proof}
Let $S=R\cup B$ be the set of vertices of a regular $2n$-gon with vertices having alternating colors. 
Let us prove that if $n-2>\max \{r_1,r_2\}> \max \{b_1,b_2\}+1$, then the configuration $(r_1,b_1,r_2,b_2)$ is not feasible for that set of points. 
If the configuration is such that $n-2>\max \{b_1,b_2\}>\max \{r_1,r_2\}+1$, the same arguments would hold switching the role of the colors. 
Without any loss of generality,  we can assume that $r_1=\max \{r_1,b_1,r_2,b_2\}$.
Observe that the hypothesis already implies $r_2 \ge 3$ and $\min \{b_1,b_2\} \ge 5$. 
We assume for contradiction that $(r_1,b_1,r_2,b_2)$ is feasible. 

Let us denote by $\hat{S}$ the boundary of the regular $2n$-gon and suppose $(r_1,b_1,r_2,b_2)$ is realized by a set $H$ of pairwise-disjoint rays.
First, observe that if a ray $h \in H$ emanating from a vertex $p \in S$ crosses the regular $2n$-gon, it splits the points of $S \setminus \{ p\} $ into two parts: the points $S^+=h^+ \cap S$ placed to the right of the ray, and the points $S^-=h^- \cap S$ placed to the left of the ray. 
The rays from the points in $S^+$ realize the positions clockwise immediately after $h$ while the rays from $S^-$ realize the positions clockwise immediately before $h$. 
We say that a ray $h \in H$ is \emph{interior} (to $\hat{S}$) if it intersects its boundary in exactly two points, and we call it \emph{exterior} otherwise.

Now let $h_1$ be the clockwise first ray of (the set of rays realizing) the block $r_1$. 
Suppose that $h_1$ is interior to $\hat{S}$. 
Then, if all the blue rays emanating from blue points of $S^+=h^+_1 \cap S$ realize the same block, it must be $b_1 \ge \lceil |S^+|/2\rceil$ and $r_1\le 1+\lfloor |S^+|/2\rfloor$ contradicting $r_1>\max \{b_1,b_2\}+1$.
If the blue rays from $S^+$ realize more than one block, then $C(S,H)$ has at least three alternations (two from $S^+$ and one from $S^-=h^-_1 \cap S$), unless $S^-$ consists of a single blue point. 
Therefore, either the first ray $h_1$ of the block $r_1$ is exterior or $S^-$ consists of a single blue point. 
In the same way we can prove that the clockwise last red ray $h_l$ realizing the block $r_1$, with origin at $p_l \in S$, either is exterior or $h_l^+ \cap S$ consists of a single blue vertex. 
We assume hereafter that the points clockwise between $p_1$ and $p_l$ are labeled in clockwise order $p_1,p_2,\ldots,p_l$.
We distinguish three cases: both rays $h_1$ and $h_l$ are exterior (case A), one of these rays is interior and the other one is exterior (case B), and both rays are interior (case C).

\paragraph{Case A}Refer to Figure~\ref{CaseA}. 
The stretch of $\hat{S}$ that goes clockwise from $p_1$ to $p_l$ contains  
\mbox{$(l-1)/2$} blue vertices and $(l-1)/2+1$ red vertices. 
Since the red rays from $h_1$ to $h_l$ realize the same block, the $(l-1)/2$ blue rays $h_2,h_4,\ldots ,h_{l-1}$, emanating from the blue vertices $p_2,p_4,\ldots ,p_{l-1}$, must be interior and also cut the boundary of $\hat{S}$ at points $\hat{p}_2,\hat{p}_4,\ldots ,\hat{p}_{l-1}$, respectively. 
In addition to the exterior rays emanating from the red points $p_1,p_3,\ldots ,p_l$, there can be at most $(l-1)/2$ additional interior red rays realizing the block $r_1$, since if more than one ray crossed the boundary of $\hat{S}$ between $p_{i}$ and $p_{i+2}$ for some $i \in \{1,3,\ldots,l-2 \}$, they would intersect each other.
Therefore, we must have $(l-1)/2+1+(l-1)/2=l \ge r_1$.

\begin{figure}[!htb]
	\centering
	\subfigure[Case A]{\includegraphics[page=1]{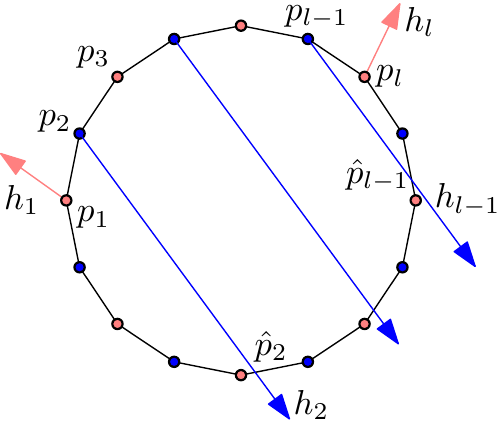} \label{CaseA}}
	\subfigure[Case B]{\includegraphics[page=2]{cases} \label{CaseB}}
	\subfigure[Case C]{\includegraphics[page=3]{cases} \label{CaseC}}
	\caption{The cases in the proof of Theorem~\ref{prop:unbalanced}.}
\end{figure}

Since the blue rays $h_2$ and $h_4$ are disjoint, there must be at least two points of $S$ counterclockwise between the points $\hat{p}_2$ and $\hat{p}_4$.
If there are exactly two points of $S$, one of them must be blue and the other one red.
The blue ray must be exterior: otherwise, it would split the block realized by $h_1$ and $h_l$.
The red ray must be interior, since otherwise clockwise between $h_4$ and $h_2$ there would be at most two red rays, contradicting that $r_2\ge 3$. 
Similarly, the ray emanating from $p_3$ must be exterior and contribute to realize the block $r_1$. 
On the other hand, between $\hat{p}_{i}$ and $\hat{p}_{i+2}$, for $i=2,\ldots, l-3$, there cannot be more than one red point. 
If there were $k\ge 2$ red points, there would be $k-1$ blue points between them. 
Since at most one of the red rays can be interior, each exterior red ray (possibly together with the ray from $p_3$) would be shot between two exterior blue rays, contradicting the fact that $r_2\ge 3$. 
Therefore, we have proved that the rays emanating from $p_2,p_4,\ldots ,p_{l-1}$ and the rays emanating from the blue vertices placed counterclockwise between the points $\hat{p}_2$ and $\hat{p}_{l-1}$ realize the same block, and there are at least $(l-1)/2+(l-1)/2-1= l-2$ of them. 
Let us denote by $H_1$ this set of blue rays. 

Observe now that, if there are blue vertices counterclockwise between $p_1$ and $\hat{p}_2$ and between $\hat{p}_{l-1}$ and ${p}_l$, some of the corresponding blue rays have to realize the same block as $H_1$ (otherwise, $H$ would have alternation number greater than two).
Hence, there would be a block of at least $l-1$ blue rays, contradicting that $l \ge r_1>\max \{b_1,b_2\}+1$. 
On the other hand, if there are no blue vertices either between $p_1$ and $\hat{p}_2$ or between $\hat{p}_{l-1}$ and $p_l$, then we have a set of consecutive vertices of $\hat{S}$ realizing a blue block and the red block $r_1$. 
But this contradicts again that $r_1>\max \{b_1,b_2\}+1$, because in a sequence of consecutive vertices the difference between the number of red and blue ones is at most one.

\paragraph{Case B} Refer to Figure~\ref{CaseB}.
Without loss of generality, assume that $h_1$ is interior and leaves only one point on its left (the blue point $p_2$). 
The case with $h_l$ interior is symmetric.

As in the previous case, the rays from the blue points $p_4,p_6,\ldots ,p_{l-1}$ must be interior and cut $\hat{S}$ in points $\hat{p}_4,\hat{p}_6,\ldots, \hat{p}_{l-1}$. 
By the same reasoning as before, between each pair of these consecutive interior blue rays there must be at least two points and exactly one of them must be red.
In addition, the red vertex counterclockwise between $\hat{p}_4$ and $\hat{p}_6$ must be the first red vertex counterclockwise from $p_1$. 
Otherwise, the ray emanating from it either would cross $h_1$ (if the ray is interior) or would create a blue block of size 2 (if the ray is exterior). By the same argument, the red vertex counterclockwise between $\hat{p}_6$ and $\hat{p}_8$ must be the second red vertex counterclockwise from $p_1$, and so on. 
This also implies that the ray from $p_3$ must be exterior since, otherwise, there would be a blue block of size at most two. 
Thus, the vertices clockwise between $p_l$ and $\hat{p}_{l-1}$ realize a single blue block and a single red one. 
As before, this contradicts the fact that $r_1>\max \{b_1,b_2\}+1$.

\paragraph{Case C} Refer to Figure~\ref{CaseC}.
In this case, since $h_1$ and $h_l$ must be disjoint, $p_l$ must be, from $p_1$, clockwise after $p_{n+1}$. 
On the other hand, since the rays from $p_4,p_6,\ldots ,p_{l-3}$ must be pairwise disjoint, $p_l$ must be clockwise before $p_{n+4}$.
Hence, $p_l=p_{n+3}$ and it is easy to see that all the red points must realize the same block and so must all the blue points, reaching again a contradiction.
\end{proof}

We next describe configurations with larger alternation number that are also not universal.

\begin{theorem}
\label{theo:no-n/3}
Let $\mathcal{C}$ be the infinite family of configurations such that any configuration $C \in \mathcal{C}$ of length $n$ has alternation number at least three, every red block of $C$ has size at least $n/k$ and every blue block of $C$ has size at least $n/l$, where $k,l \in \R$. 
Then, there exists $n_0 \in \N$ such that any $C \in \mathcal{C}$ of size $n > n_0$ is not universal. 
In particular, the uniform configuration $(n/k,n/k,...,n/k,n/k)$ with $k \geq 3$ and $n/k \in \N$ is not universal for large enough $n$.
\end{theorem}

Before proving the previous theorem, we need two technical lemmas, which we formalize next.

Given a real number $\lambda>0$, let $K_\lambda(n)$ be the set of $n$ (complex) roots of the unity, taken as points in the real plane, and scaled by a factor of $\lambda$.
The \textit{width} of a point set $T$, is the width of the thinnest slab (closed space between two parallel lines) enclosing $T$. 
The width of a slab is the distance between its boundary lines. 
We say that a slab $Z$ \textit{certifies} the width $w$ of a point set $T$ if $Z$ has width $w$ and $T \subset Z$.

\begin{figure}[t]
	\centering
\includegraphics[scale=1,page=5]{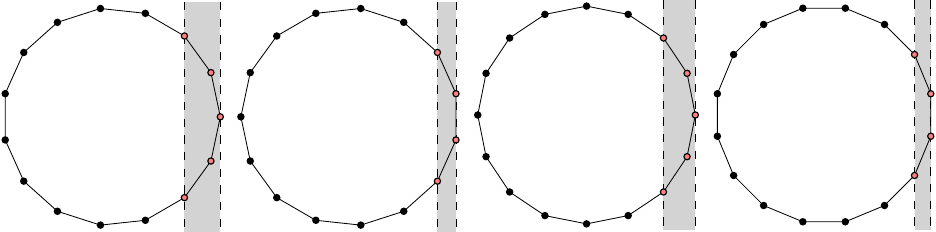}
\caption{Narrowest slab for odd and even $n$, and odd and even $t$.}
\label{coord}
\end{figure}

\begin{lemma} 
\label{minwidth}
The width of any set $T \subset K_1(n)$ with $3 \leq \vert T\vert  \leq \lfloor n/2 \rfloor$ is at least $$\cos\left( \frac{\pi}{n}\right)  - \cos \left( \frac{(\vert T\vert-1)\pi}{n}\right) .$$
\end{lemma}

\begin{proof}
Let $T \subset K_1(n)$ be a set of minimum width. 
If $Z$ certifies the  width of $T$, then each of the bounding lines of $Z$ must contain at least one point of $T$ (otherwise, the slab can obviously be made narrower). 
Furthermore, at least one of the lines must contain two points of $T$.
Otherwise, if every bounding line contains only one point of $T$, say $p$ and $q$ respectively, rotating the lines in a parallel fashion pivoting on $p$, respectively on $q$, would lead to a thinner slab  enclosing $T$.
Assuming that the slab is vertical (if not, we can first apply a suitable rotation),  
the abscissae of $K_1(n)$ are, without loss of generality, either

%\begin{linenomath}
\begin{align*} 
\mathcal{A}(n)=&\left\lbrace \cos\left( (2 \pi (j-1))/ n \right)  \mid j \range{ \left\lceil (n+1)/2 \right\rceil} \right\rbrace, \; or \\
\mathcal{A}'(n)=&\left\lbrace \cos\left((2 \pi (j-1) + \pi )/ n \right) \mid j \range{ \left\lceil n/2 \right\rceil} \right\rbrace.
\end{align*}
%\end{linenomath}

Simple trigonometric calculations show that the sets having the rightmost abscissae of $\mathcal{A}(n)$ if $\vert T\vert$ is odd, and of $\mathcal{A}'(n)$ if $\vert T\vert$ is even, have minimum width.  
These point sets are sketched in Figure~\ref{coord} and their widths are
%\begin{linenomath}
\begin{align*} 
& 1 - \cos((\vert T\vert-1)\pi / n ), \; \text{ if $\vert T\vert$ is odd, and } \\
& \cos(\pi / n ) - \cos((\vert T\vert-1)\pi / n ), \;  \text{ if $\vert T\vert$ is even.}
\end{align*}
%\end{linenomath}
This completes the proof.
\end{proof}

Given a point $p$ outside the unit disk centered at the origin, we define $V_p$ to be the open wedge defined by the rays starting at $p$ and tangent to the unit circle, and containing the origin. 
We refer to Figure~\ref{fig:lemmaVp} for an illustration of the following lemma.

\begin{lemma}
\label{wedge}
For any $p \in K_\lambda(n)$ with $\lambda>1$, it holds \[\vert (K_\lambda(n) \cap V_p)  \vert \leq \frac{2n}{\pi} \arcsin\left( \frac{1}{\lambda}\right)+1.\]
\end{lemma}

\begin{proof}
The angle of $V_p$ is $2  \arcsin\left( \frac{1}{\lambda}\right)$. Consider the circle in which $K_\lambda(n)$ is inscribed. 
This is cut by $V_p$ in a circular arc. 
This arc is seen from the origin with an angle of $4 \arcsin\left( \frac{1}{\lambda}\right)$. Since $K_\lambda(n)$ is a regular polygon, two consecutive points are seen from the center with an angle of $2 \pi /n$. Therefore, at most $1+4 \arcsin\left( \frac{1}{\lambda}\right) \frac{n}{2\pi}$ points of $K_\lambda(n)$ can lie in the arc.
\end{proof}

\begin{figure}[!htb]
	\centering
	\hfill
	\subfigure[]{\includegraphics[scale=1,page=3]{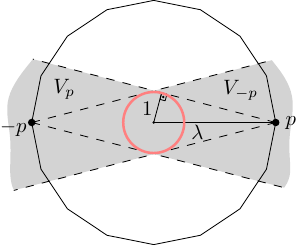} \label{fig:lemmaVp}}
	\hfill
	\subfigure[]{\includegraphics[scale=1,page=4]{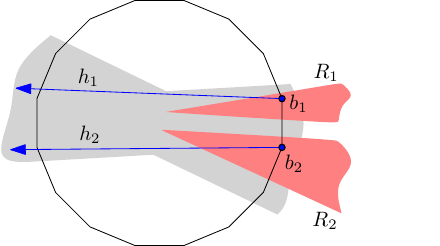}\label{constrNoUniv}}
	\hfill
	\caption{Drawings for Lemma~\ref{wedge} and Theorem~\ref{theo:no-n/3}}
\end{figure}

We prove now Theorem~\ref{theo:no-n/3}; see Figure~\ref{constrNoUniv}.

\bigskip
\noindent {\bf Proof of Theorem~\ref{theo:no-n/3}.}
Note first that it must be the case that $k,l \geq 3$ since there are at least three blocks of each color.
Let $R=K_1(n)$ and $B=K_\lambda(n)$ with $\lambda > 1$, and $\mathcal{C}_R$ and $\mathcal{C}_B$ be the circles containing $R$  and $B$, respectively.
Using Lemma~\ref{minwidth}, it is easy to see that if $\lambda$ is smaller than 
\[g(n)=\left[\cos\left(\frac{\pi}{n} \right)-\cos\left(\frac{(  \lceil \frac{n}{k} \rceil-1) \pi}{ n} \right) \right] \left[ 2 \sin\left(\frac{\pi}{n}\right) \right]^{-1},\]
the rays emanating from a subset $R_1 \subset R$ realizing a block of the configuration will have to cross at least two arcs of $\mathcal{C}_B$ between points of $B$, since the first factor is a lower bound for the width of $R_1$ and the distance between two consecutive points of $B$ is $2 \lambda \sin(\pi/n)$.
Therefore, the ray from at least one point $b_1 \in B$ will have to intersect $\mathcal{C}_R$ because, otherwise, it would split the block of $R_1$. Let $b_1, b_2, b_3 \in B$ be points belonging to three different sets $R_1,R_2,R_3 \subset R$ realizing each of them a red block.
Note now that it has to be $b_2,b_3 \in V_{b_1} \cup V_{-b_1}$, where $-b_1$ indicates the point in $\mathcal{C}_B$ symmetric to $b_1$ with respect to the origin, since the rays $h_2$ placed at $b_2$ and $h_3$ placed at $b_3$ should split $R$ and they should not intersect the ray $h_1$ placed at $b_1$.
Observe that either $V_{b_1}$ or $V_{-b_1}$ must contain at least two of the points $b_1$, $b_2$ and $b_3$.
Assume these to be $b_1$ and $b_2$. 
Note then that only the points from the two arcs of $\mathcal{C}_B$ between $h_1$ and $h_2$ can realize a blue block between $R_1$ and $R_2$.
With the help of Lemma~\ref{wedge} to bound the number of points of $B$ in these arcs, we have that if $\lambda$ is larger than 
\[f(n)= \left[ \sin \left( \frac{\pi}{4n} \left(   \left  \lceil \frac{n}{l} \right \rceil -2 \right) \right) \right]^{-1},\]
no block of $B$ can be realized between $R_1$ and $R_2$.
Thus, for $n$ and $\lambda$ such that $1 < f(n)< \lambda < g(n),$ the configuration is not feasible.
Since $g(n) \rightarrow \infty$ and \mbox{$f(n) \rightarrow [\sin(\pi/4l)]^{-1}>1$}, the counterexample can be certainly constructed. \hfill $\square$

\subsection{Deciding feasibility of configurations} 
\label{subsec:deciding}

In this section we study algorithms to decide if a given configuration can be realized for a given point set. 
We start by studying the case of points on a line, and then we focus on the case of points in general position.

The following algorithm is an adaptation of an algorithm by Akiyama and Urrutia~\cite{AU907} for deciding, given $n$ red and $n$ blue points on a circle, whether they admit a simple Hamiltonian polygonal path in which the colors of the vertices alternate.

\begin{theorem}
Given a bichromatic point set $S=R\cup B$ on a line~$\ell$ and a
configuration $C$, it can be decided in $O(n^2)$ time whether $C$ is feasible for~$S$.
\end{theorem}

\begin{proof}
Without loss of generality, assume that $\ell$ is horizontal, and let $S=\{ p_1,\dots,p_{2n} \}$, where the indices are taken from left to right.
Note that any realization from $S$ can be perturbed such that all the rays are vertical.
The point $p_1$ must realize a position of its same color in $C$. 
Then, $p_2$ will realize either the previous position of the configuration or the next one, depending on whether the corresponding ray is pointing downwards or upwards.
One, two or none of the previous options will be valid depending on whether the color of the previous and next positions of $C$ match the color of $p_2$.
In this way, when we traverse $S$ from left to right, choosing the upward or the downward ray for each point, we may be realizing a subsequence of consecutive elements in $C$.

Consider the directed graph having a node for each one of the $\Theta(n^2)$ subsequences of $C$. 
Note that we consider as different two equal red-blue patterns if they start at different positions of $C$. 
We add an arc from a node corresponding to a subsequence of length $k \geq 0$ to a node corresponding to a subsequence of length $k+1$ if the second subsequence can be obtained from the first one by attaching the color of $p_{k+1}$ before or after it.
It is clear that a configuration is feasible for $S$ if and only if there exists a path from the empty sequence to some of the $2n$ linear subsequences of $C$ of length $2n$ in the aforementioned directed graph.
Since the out-degree of every node is at most $2$, the size of the graph is quadratic and the decision can be made in $O(n^2)$ time.
\end{proof}

For the general setting, the decision question can also be answered in polynomial time. Next, we describe a polynomial algorithm to decide whether a configuration $C$ is feasible from a given point set $S$. We explain the algorithm assuming that $S=\{p_1,\ldots,p_n \}$ is a point set in (strong) general position and $C$ is a cyclic permutation of $[n]$. The algorithm decides in $O(n^{11})$ time and $O(n^9)$ space if $C$ is feasible from $S$, using pairwise-disjoint rays, and it works in a similar manner, regardless the number of colors used to color the points and the number of points of each color.  We do not give all the details of the algorithm, but only the main ideas. The rest of this section is devoted to prove the following theorem.

\begin{theorem}
	\label{prop:convexpos}
	Given a point set $S$ in general position and a configuration $C$, it can be decided in polynomial time whether $C$ is feasible from $S$.
\end{theorem}

\subsubsection{Rays in canonical position}

First, we state two lemmas from~\cite{GHTU07} and some definitions that will be useful.

\begin{lemma}{(\cite{GHTU07})}\label{lem:canonical}
Let $S=\{p_1,\ldots,p_n \}$ be a point set and $C$ be a cyclic permutation of $[n]$.
There is a set $H$ of pairwise-disjoint rays from $S$ realizing~$C$ if and only if there is a set $H'$ of non-crossing rays from $S$, having direction vectors in $V= \{ p-q \ | \ p, q \in S, p \neq q \}$, that realize $C$ as well.
\end{lemma}

The set $H'$ of rays in the previous lemma is obtained from $H$ by rotating each ray clockwise until it hits a point of $S$ or it becomes parallel to another ray (and they continue rotating in a parallel manner).
The non-crossing rays resulting from the aforementioned perturbation are said to be in {\it canonical position} (see Figure \ref{fig:separable} for examples). 
The configuration induced by a set of non-crossing rays in canonical position is the same as the one induced by a set of rays resulting from rotating counterclockwise every ray a small angle, which is a set of pairwise-disjoint rays.
 
In view of the previous lemma, to decide whether a configuration $C$ can be realized or not, we can restrict ourselves to study sets of rays with direction vector in $V$, that is, sets of rays in canonical position.

\begin{figure}[!htb]
	\centering
	\includegraphics[scale =1]{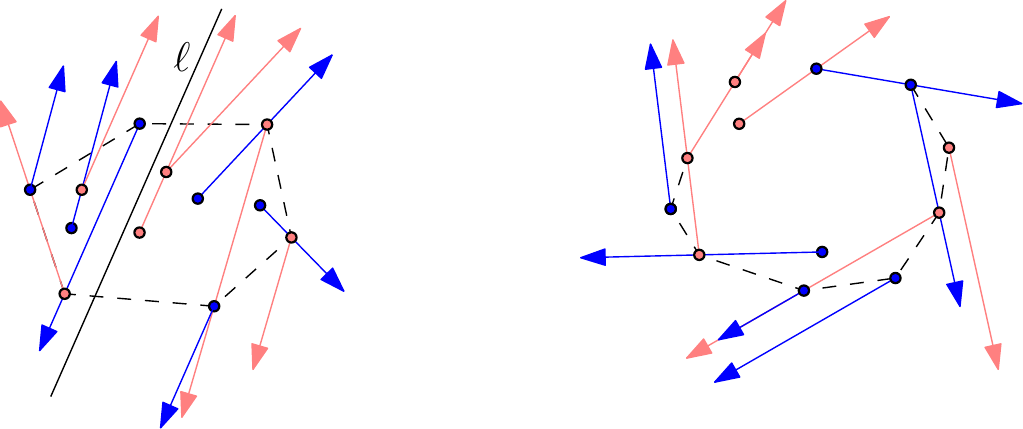}
	\caption{Separable and non-separable sets of rays in canonical position.}\label{fig:separable}
\end{figure}

A set of rays in canonical position is {\it separable} when there exists a line $\ell$ that does not intersect any ray. If such a line does not exists, we say that the set is {\it non-separable}. Note that, in a non-separable set, the extension of any ray $h$ in the opposite direction always hits another ray $h'$. Otherwise, we could take the line supporting $h$, infinitesimally translated, for a separator. In Figure \ref{fig:separable}, a separable (left) and a non-separable (right) set of rays in canonical position are shown.

\begin{lemma}{(\cite{GHTU07})}\label{lem:decomp}
	Let $H=\{h_1, \ldots , h_n\}$ be a set of non-crossing rays $H$ in canonical position from $S$, where $h_i$ has apex $p_i\in S$ for all $i\in [|S|]$.
	If $H$ is non-separable, then one of the following statements holds.
	\begin{itemize}
		\item[(i)] There are three points $p_i,p_j,p_k \in S$, in clockwise order (in their convex hull $\triangle p_ip_jp_k$), such that no ray crosses $\triangle p_ip_jp_k$, and the clockwise angles defined by consecutive rays emanating from them are less than $\pi$ (see Figure~\ref{fig:non-separable} (left) for an example).
		\item[(ii)] There are three points $p_i,p_j,p_k \in S$, in clockwise order (in their convex hull $\triangle p_ip_jp_k$), such that the clockwise angles defined by the consecutive rays emanating from them are less than $\pi$, and no ray crosses the quadrilateral $q_iq_jp_jp_k$, where $q_j$ is the crossing point between $h_j$ and the extension of $h_i$ in the opposite direction, and $q_i$ is the crossing point between the extensions of $h_i$ and $h_k$ in the opposite directions (see Figure~\ref{fig:non-separable} (right) for an example).
	\end{itemize}
\end{lemma}

\begin{figure}[!htb]
	\centering
	\includegraphics[scale=1]{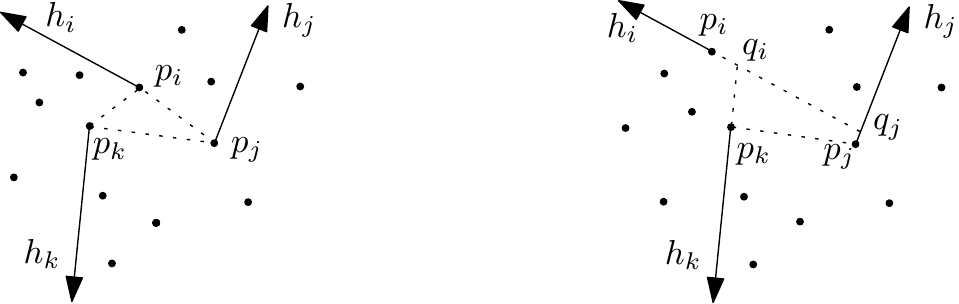}
	\caption{Illustration of Lemma \ref{lem:decomp}.}\label{fig:non-separable}
\end{figure}

In the next subsection, we define formally the types of regions that rays $h_i,h_j,h_k$ and segments $p_ip_j, p_jp_k, p_kp_i$ of the previous lemma define.

\subsubsection{$\Pi$-feasible, $\Sigma$-feasible and $\Lambda$-feasible tuples}

Given a configuration $C$, a {\it subconfiguration} of $C$ is a subsequence of $C$ formed by (cyclically) consecutive positions of $C$. Given $p_i,p_j \in S$, $p_i \ne p_j$, and $u,v \in V$, we say that the tuple $(p_i,p_j,u,v)$ is a \emph{$\Pi$-tuple} if $h(p_i,u)$ and $h(p_j,v)$ are non-crossing and the clockwise angle between $u$ and $v$ is less than $\pi$.
Given a $\Pi$-tuple $(p_i,p_j,u,v)$, let $S(p_i,p_j,u,v)$ be the set of points of $S$ contained in the region $R(p_i,p_j,u,v)$ bounded by $h(p_i,u)$ (included), $h(p_j,v)$ (excluded) and the segment $p_ip_j$, which contains a ray with direction $(u+v)/2$. See Figure~\ref{fig:tipo1}.
Given a subconfiguration $C'$ of $C$, we say that a tuple $(p_i,p_j,u,v,C')$ is \emph{$\Pi$-feasible} if $C'$ can be realized by a set of rays $H'$ in canonical position from $S(p_i,p_j,u,v)$ where each ray of $H'$ is contained in $R(p_i,p_j,u,v)$. Note that $p_i$ belongs to $S(p_i,p_j,u,v)$ but $p_j$ does not.

\begin{figure}[!htb]
	\centering
	\includegraphics[scale=1]{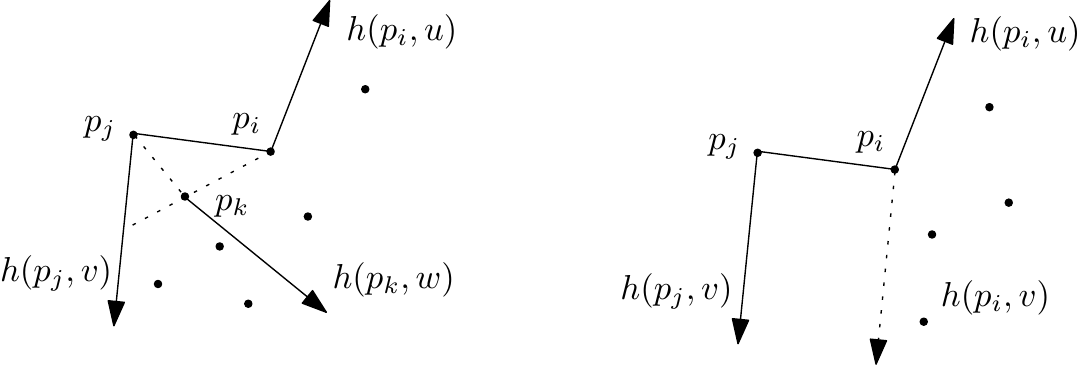}
	\caption{Decomposing a $\Pi$-tuple $(p_i,p_j,u,v)$ into two $\Pi$-tuples consisting of fewer points (left) or into one $\Lambda$-tuple (right).} \label{fig:tipo1}
\end{figure}

Given $p_i,p_j \in S$, $p_i \ne p_j$, and $u,v,w \in V$, we say that the tuple $(p_i,p_j,u,v,w)$ is a \emph{$\Sigma$-tuple} if $h(p_i,u)$ and $h(p_j,w)$ are non-crossing, the clockwise cyclic order of $h(p_i,u), h(p_i,v), h(p_j,w)$ is $h(p_i,u), h(p_j,w), h(p_i,v)$, the clockwise angle between $u$ and $v$ is less than or equal to $\pi$, and the ray $h(p_i,v)$ crosses the line supporting $h(p_j,w)$.
Given a $\Sigma$-tuple $(p_i,p_j,u,v,w)$, let $S(p_i,p_j,u,v,w)$ be the set of points of $S$ contained in the (convex) region $R(p_i,p_j,u,v,w)$ bounded by $h(p_i,u)$ (included), the line supporting $h(p_j,w)$ (excluded), and $h(p_i,v)$, which contains a ray with direction $(u+w)/2$. See Figure~\ref{fig:tipo2}.
Given a subconfiguration $C'$ of $C$, we say that a tuple $(p_i,p_j,u,v,w,C')$ is \emph{$\Sigma$-feasible} if $C'$ can be realized by a set of rays $H'$ in canonical position from $S(p_i,p_j,u,v,w)$ where each ray of $H'$ is contained in $R(p_i,p_j,u,v,w)$.

\begin{figure}[!htb]
	\centering
	\includegraphics[scale=1]{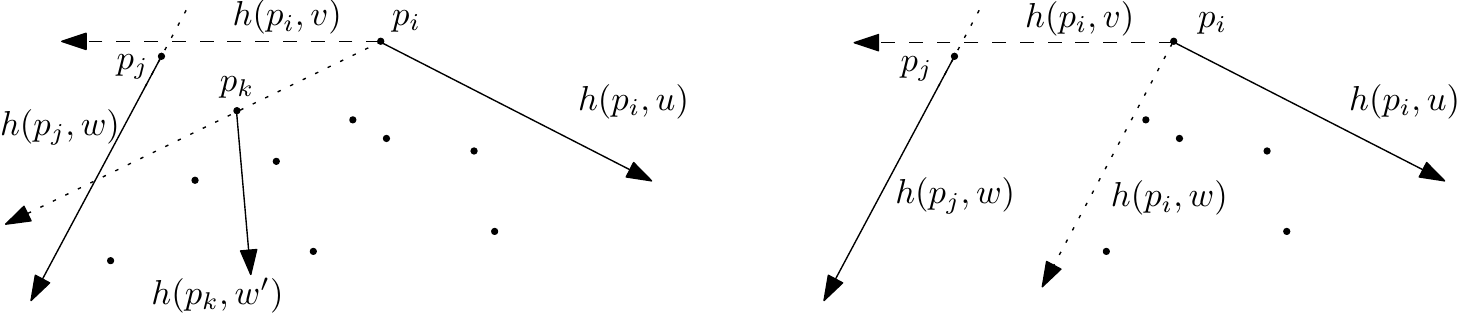}
	\caption{Decomposing a $\Sigma $-tuple $(p_i,p_j,u,v,w)$ into a $\Pi$-tuple and a $\Sigma$-tuple consisting of fewer points (left), or into one $\Lambda$-tuple (right).}\label{fig:tipo2}
\end{figure}

If in the definition of $\Pi$-tuples we allowed $p_i=p_j$, we would obtain what we call $\Lambda$-tuples, defined formally as follows.
Given $p_i \in S$ and $u\ne v \in V$, we say that the tuple $(p_i,u,v)$ is a \emph{$\Lambda$-tuple} if the clockwise angle between $u$ and $v$ is less than $\pi$.
Given a $\Lambda$-tuple $(p_i,u,v)$, let $S(p_i,u,v)$ be the set of points of $S$ contained in the (convex) region $R(p_i,u,v)$ bounded by $h(p_i,u)$ (included) and $h(p_i,v)$ (excluded), which contains a ray with direction $(u+v)/2$. See Figure~\ref{fig:tipo3}.
Given a subconfiguration $C'$ of $C$, we say that a tuple $(p_i,u,v,C')$ is \emph{$\Lambda$-feasible} if $C'$ can be realized by a set of rays $H'$ in canonical position from $S(p,u,v)$ where each ray from $H'$ is contained in $R(p_i,u,v)$.

\begin{figure}[!htb]
	\centering
	\includegraphics[scale=1]{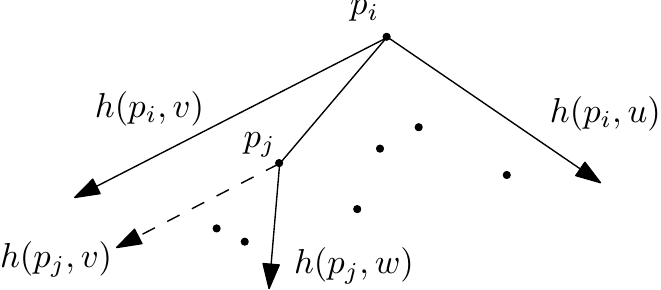}
	\caption{Decomposing a $\Lambda$-tuple $(p_i,u,v)$ into a $\Pi$-tuple and a $\Lambda$-tuple.}\label{fig:tipo3}
\end{figure}

Even though $\Lambda$-tuples can be seen as a particular case of more general $\Pi$-tuples, we treat them separately because we use different methods to decide the feasibility of a subconfiguration from the point sets associated with the tuples.

\subsubsection{Algorithm}

We present a dynamic programming algorithm to determine whether a given configuration can be realized from $S$ by a set of rays in canonical position. In a first step, we will maintain three tables $T_\Pi$, $T_\Sigma$, and $T_\Lambda$ in which all $\Pi$-, $\Sigma$- and $\Lambda$-feasible tuples, respectively, will be stored.
More formally, the table $T_\Pi$ consists of entries of the form  $(p_i,p_j,u,v,C')$, where $(p_i,p_j,u,v)$ is a $\Pi$-tuple and $C'$ is a subconfiguration of $C$ of length $|S(p_i,p_j,u,v)|$, and stores whether the tuple $(p_i,p_j,u,v,C')$ is $\Pi$-feasible.
The tables $T_\Sigma$ and $T_\Lambda$ are defined analogously.

In the following, we explain how to compute the values in these tables recursively.
That is, we show how to check the feasibility of a tuple using the feasibility of some tuples with fewer points. In order to check the $\Pi$-feasibility of $(p_i,p_j,u,v,C')$, we can sweep counterclockwise the ray $h(p_i, p_j-p_i)$ around its apex $p_i$ until we hit a point $p_k \in S(p_i,p_j,u,v)$. If $h(p_i, p_k-p_i)$ crosses $h(p_j,v)$, then we only need to check the $\Pi$-feasibility of $(p_i, p_k, u, w,C'_1)$ and $(p_k,p_j,w, v,C'_2)$ for each $w\in V$ such that the ray $h(p_k,w)$ is contained in $R(p_i,p_j,u,v)$, where $C'_1$ consists of the first $|S(p_i,p_k,u,w)|$ positions of $C'$ and $C'_2$ consists of the last $|S(p_k,p_j,w,v)|$ positions of $C'$ (see Figure \ref{fig:tipo1}, left). 
If $h(p_i, p_k-p_i)$ does not cross $h(p_j,v)$, then we only need to check the $\Lambda$-feasibility of the tuple $(p_i,u, v,C')$ (see Figure \ref{fig:tipo1}, right).

Similar analysis can be done to check the $\Sigma$-feasibility of the tuple $(p_i,p_j,u,v,w,C')$ by counterclockwise sweeping the ray $h(p_i,v)$ around $p_i$ until a point $p_k$ is hit; and to check the $\Lambda$-feasibility of the tuple $(p_i,u,v,C')$ by sweeping $h(p_i,v)$ around $p_i$ until a point $p_j\in S(p_i,u,v)$ is hit.
See Figures \ref{fig:tipo2} and \ref{fig:tipo3}, respectively.

The sizes of $T_\Pi$, $T_\Sigma$ and $T_\Lambda$ are $O(n^7), O(n^9)$ and $O(n^6)$, respectively, and each entry can be computed in $O(n^2)$ time (the feasibility must be checked for every vector in $V$, in the worst case).
Therefore, the tables can be constructed incrementally in $O(n^9), O(n^{11})$ and $O(n^8)$ time, respectively, by interleaving the calculations between them.

Using these tables, we test the feasibility of $C$ as follows.
We can check in $O(n^6)$ time whether $C$ can be realized by a separable set of rays in canonical position from $S$.
For any line $\ell$ with direction vector $u\in V$ and leaving $k$ points of $S$ to its left (there are $O(n^3)$ choices for $\ell$) 
and for any partition of $C$ into two disjoint subconfigurations $C_1$ and $C_2$ of sizes $k$ and $n-k$ ($O(n)$ choices), respectively, we only need to check the feasibility of $C_1$ and $C_2$ for $O(n^2)$ $\Lambda$-tuples.
For instance, assuming that $\ell$ is horizontal, we choose the point $p_i$ with highest ordinate below $\ell$ and explore the different subproblems defined by the rays $h(p_i,u), h(p_i,-u)$ and $h(p_i,v)$, with $v\in V$ such that the ray
$h(p_i,v)$ is below $\ell$.
See Figure \ref{fig:separabledec}.

\begin{figure}[!htb]
	\centering
	\includegraphics[scale=1]{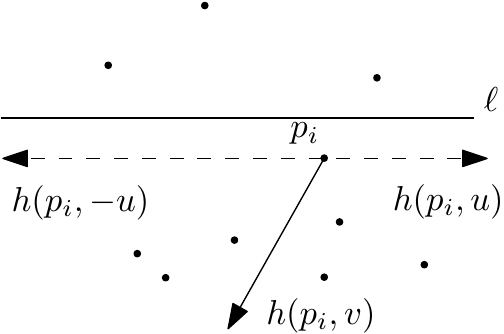}
	\caption{Decomposing a separable set.}\label{fig:separabledec}
\end{figure}

On the other hand, as a consequence of Lemma~\ref{lem:decomp}, we can decide in $O(n^{10})$ time whether $C$ can be realized using a non-separable set of rays in canonical position.
For instance, if part $(i)$ of that lemma holds, then, using the information stored in $T_\Pi$, for every triple of points $p_i,p_j,p_k \in S$, for every triple of vectors $u,v,w \in V$ such that the angles from $u$ to $v$, from $v$ to $w$ and from $w$ to $u$ are less than $\pi$, and for every partition of $C$ into  $C_1, C_2, C_3$ of appropriate sizes, we need to check only the $\Pi$-feasibility of the tuples $(p_i, p_j,u, v,C_1)$, $(p_j,p_k,v, w,C_2)$ and $(p_k,p_i,w,u,C_3)$.
A similar analysis can be done if part $(ii)$ of the lemma holds.

This concludes the proof of Theorem~\ref{prop:convexpos}.

\section{Full-crossing sets of rays}\label{sec:nondisjoint}

In this section, we study sets of rays that are not necessarily disjoint. 
In this scenario, it will be very useful to consider, given a ray $h$ with apex $p$, the ray
$\overline{h}$ with apex $p$ oriented in the direction opposed to that of $h$,
which we call the \emph{inverse} ray of $h$.
We also call a set $H$ of rays \emph{pairwise-proper} if no ray is contained in another one. 

It is not surprising that by removing the disjointness constraint
all configurations are realizable by a proper set of rays, no matter the position of the points in $S$. 

However, there exist point sets $S=R\cup B$, with $|R|=|B|=n$, such that any proper set $H$ of rays from~$S$ realizing a certain configuration produces $\Theta(n^2)$ crossings.

\begin{theorem} \label{thm:nsquarecrossings}
There exist point sets $S=R\cup B$ such that any set $H$ of pairwise-proper rays from $S$ realizing the configuration $C=(r,b,r,b, \ldots , r,b)$ has $\Theta (n^2)$ crossings.
\end{theorem}

\begin{proof}
	For clarity of exposition, we assume that $n$ is even. 
	Consider the set $R=\{(-i,0) | i \in [n] \}$ on the $x$-axis $OX$ (see Figure~\ref{fig:nsquared}). 
	Consider now the points $p=(0,1)$ and $p'=(0,-1)$ and place half of the elements of $B$ in the interior of a disk $\mathcal{C}$ of radius $\epsilon<1/2$ centered at $p$, and the other half in the interior of a disk $\mathcal{C}'$ of radius $\epsilon$ centered at $p'$.
	Perturb $S$ slightly such that the points lie in general position, no point lies on the $OX$ or $OY$ axis, and no pair of points defines a vertical or horizontal line.

	\begin{figure}[thb!]
		\centering
		\includegraphics{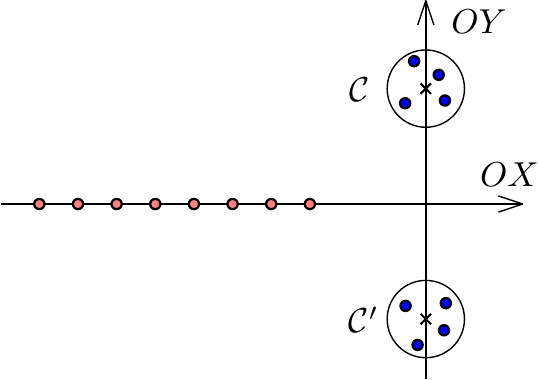}
		\caption{Any set $H$ of rays realizing $(r,b,r,b, \ldots , r,b)$ has at least $\Theta(n^2)$ crossings.}
		\label{fig:nsquared}
	\end{figure}
	
	Suppose that $H$ is a set of rays from $S$ realizing $C=(r,b,r,b, \ldots , r,b)$. 
	Let us see that $H$ has $\Theta (n^2)$ crossings. 
	We divide the set of rays with apices in $\mathcal{C}\cap B$ into the three sets
	$H_1$, $H_2$, and $H_3$. 
	Let $H_1$ be the set of rays with direction vector contained in the right half-plane (defined by~$OY$), $H_2$ the set of rays with direction vector contained in the left half-plane and that do not cross~$OX$, and $H_3$ the set of rays with direction vector in the left half-plane that do cross~$OX$. 
	Let us define $n_1 =|H_1|$, $n_2 =|H_2|$ and $n_3 =|H_3|$. 
	The sets $H'_1, H'_2$ and $H'_3$ and the cardinalities $n'_1, n'_2$ and $n'_3$ are defined similarly for the rays with apices in $\mathcal{C}'\cap B$.
	Since $n_1+n_2+n_3 =n/2 $ and $n'_1+n'_2+n'_3=n/2 $, necessarily at least one of $n_1, n_2$ and $n_3$ and one of $n'_1, n'_2$ and $n'_3$ are larger than or equal to $n/6$.  
	
	Suppose that $n_1 \ge n/6$, and assume that the labels in $H_1=\{h_1,\ldots, h_{n_1}\} $ correspond to the order on the slope of the rays.
	As $H$ realizes $(r,b,r,b,\ldots , r,b)$, at infinity between two angularly consecutive rays $h_i,h_{i+1} \in H_1$, there must appear at least one red ray.  
	Let $G\subset H$ be a set of red rays interleaved with the rays of $H_1$ such that exactly one red ray appears between two consecutive rays of $H_1$. Let $G_1 \subset G$ be the rays that intersect $\mathcal{C}$, and $G_2 \subset G$ be the rays that do not intersect $\mathcal{C}$.  
	At least one of $|G_1|$ and $|G_2|$ is larger than or equal to~$n/12$. 
	If $|G_1| \ge n/12$, then $H$ has $\Theta(n^2)$ crossings because two red rays passing through $\mathcal{C}$ must cross each other. 
	If $|G_2| \ge n/12$, observe that, the red ray of $G_2$ appearing between the two consecutive rays of $h_i,h_{i+1} \in H_1$ necessarily crosses the first $i$ blue rays of $H_1$ or the last $n-i$ rays of $H_1$, depending on whether it leaves $\mathcal{C}$ to the right or to the left. 
	Therefore, since either the rays of $G_2$ that leave $\mathcal{C}$ to its left, or the ones that leave it to its right, are $\Theta(n)$, $H$ has $\Theta(n^2)$ crossings.
	
	Similarly, if $n_2 \ge n/6$, the red rays interleaving the rays in $H_2$ generate $\Theta(n^2)$ crossings.
	If $n'_1$ or $n'_2$ are at least $n/6$, symmetric arguments lead to the same conclusion. 
	If none of $n_1,n_2,n_1'$ and $n_2'$ is at least $n/6$, then $n_3 \ge n/6$ and $n'_3 \ge n/6$ and $H$ has at least $\Theta (n^2)$ crossings because each ray of $H_3$ crosses each ray of $H'_3$.
\end{proof}

It is quite likely that, given $S=R\cup B$ and a configuration $C$, finding a pairwise-proper set of rays from $S$ that realizes $C$ and minimizes the number of crossings is an NP-hard problem.
Yet we have no proof of that.

The rest of the section concerns the study of configurations that can be realized by sets of rays that pairwise cross, which we call \emph{full-crossing} sets of rays.
We say that a configuration $C$ is {\em $\chi$-feasible} (from a point set $S$) if there exists a set $H'$ of full-crossing rays from $S$ such that $C$ is equal to $C(S,H')$.
Observe that if $H$ is full-crossing, then the elements of the set $\overline{H}$ of
inverse rays are pairwise-disjoint and realize the same configuration as $H$.

We show first that, as in the case of pairwise-disjoint rays, there exist universal point sets for full-crossing rays. 
Then, we prove several results showing that the full-crossing scenario is significantly different that the pairwise-disjoint case. 

\begin{theorem}\label{fcuniv}
	For every $n \in \N$, there exists a bichromatic point set $S=R\cup B$ such that every configuration is $\chi$-feasible from $S$.
\end{theorem}

\begin{proof}
	For any given $C$, consider the point set $S$ and the set $H$ of rays from $S$ described in the proof of Theorem~\ref{theo:universal-point-set}.
	Since $C(S,H)=C$, the set $\overline{H}=\{ \bar{h} | h\in H \}$ realizes $C$ as well. 
	Let $L$ be the set of supporting lines of the rays in $H$. 
	Note that no two lines of $L$ intersect above the $x$-axis and they are pairwise nonparallel and, thus, they pairwise intersect below $OX$. 
	Since no ray of $H$ intersects the halfplane below $OX$, every ray in $\overline{H}$ intersects every line in $L$ and, hence, $\overline{H}$ is full-crossing.
\end{proof}

In contrast with the case of pairwise-disjoint rays, we will prove that there are no universal configurations for the full-crossing scenario. 
First, we prove a lemma in this direction.

\begin{lemma}\label{lem:hex}
Let $S$ be a set of six alternating red and blue points being the vertices of a convex hexagon. Then
the configuration $(r,r,r,b,b,b)$ is not $\chi$-feasible from $S$.
\end{lemma}

\begin{proof}
Let $S=\{p_1,\ldots,p_6\}$ be the vertices of the hexagon, listed counterclockwise, and let us assume that $p_1,p_3,p_5$ are red and that $p_2,p_4,p_6$ are blue
(see Figure \ref{fig:hexagono1}). We denote by $e_i$ the edge $p_ip_{i+1}$ of the hexagon, with arithmetic of indices modulo $6$.

\begin{figure}[!htb]
	\centering
	\subfigure[~]{\includegraphics[page=1]{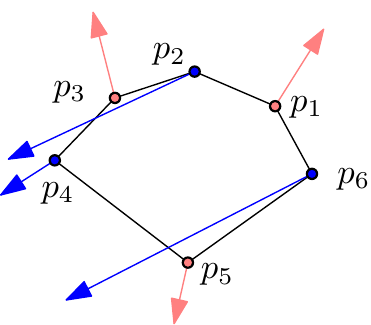} \label{fig:Case1}}
	\subfigure[~]{\includegraphics[page=2]{hexagono1b} \label{fig:Case2}}
	\subfigure[~]{\includegraphics[page=3]{hexagono1b} \label{fig:Case3}}
	\subfigure[~]{\includegraphics[page=4]{hexagono1b} \label{fig:Case4}}
	\caption{The cases in the proof of  Lemma \ref{lem:hex}. For clarity, some rays are not shown.}\label{fig:hexagono1}
\end{figure}

Assume for contradiction that there is a full-crossing set of rays from $S$ that realizes the configuration $(r,r,r,b,b,b)$. 
Let us denote their inverse rays by $h_1,\ldots,h_6$, where each $h_i$ has apex at $p_i$. 
Both the set of non-crossing rays $h_1,\ldots,h_6$ and their full-crossing inverse rays $\bar{h}_1,\ldots,\bar{h}_6$ realize the configuration $(r,r,r,b,b,b)$. 

Since the rays $\bar{h}_1,\ldots,\bar{h}_6$ cross pairwise, no two of them are parallel, and this also applies to their inverses $h_1,\ldots,h_6$. 
Therefore, by infinitesimal perturbation, we can assume that none of these rays intersects the boundary of the hexagon in more than two points.

If $h_1$, $h_3$ and $h_5$ are exterior, the configuration $(r,r,r,b,b,b)$ can only be achieved if two of the blue rays, say $h_2$ and $h_6$, are interior, and cross $e_3\cup e_4$ (Figure~\ref{fig:Case1}). 
But in this case, $\bar{h}_1\cap \bar{h}_2=\varnothing$ or $\bar{h}_1\cap \bar{h}_6=\varnothing$, a contradiction. 
Therefore, the three red rays cannot be all exterior, and neither can be, by symmetry, all three blue rays.

If two red rays are interior, say $h_1$ and $h_3$, then $h_2$ must be exterior (Figure~\ref{fig:Case2}), but then $\bar{h}_2\cap \bar{h}_1=\varnothing$ or $\bar{h}_2\cap \bar{h}_3=\varnothing$, a contradiction.

Therefore, there is only one possibility left: for each of the colors, two rays  must be exterior and one interior. Let us assume, for example, that $h_2$ is interior while $h_4$ and $h_6$ are exterior. 
To achieve the configuration $(r,r,r,b,b,b)$ the ray $h_5$ must be the only interior red  ray. 
If  $\bar{h}_5\cap \bar{h}_2=\varnothing$, as in Figure~\ref{fig:Case3}, we get a contradiction, therefore we have to assume that  $\bar{h}_5\cap \bar{h}_2\ne\varnothing$.
However, in this case $h_1$ intersects $\bar{h}_2$ or $\bar{h}_5$ (Figure \ref{fig:Case4}), and hence
$\bar{h}_1\cap \bar{h}_2=\varnothing$ or $\bar{h}_1\cap \bar{h}_5=\varnothing$, reaching another contradiction.
\end{proof}

\begin{theorem}\label{prop:nouniv}
	There exist no universal configurations $C$ of length $2n \ge 20$ for full-crossing rays. 
	That is, for every configuration $C$ of length at least $20$ there exists a bichromatic point set $S=R\cup B$ such that
	$C$ is not $\chi$-feasible from $S$.
\end{theorem}

\begin{proof}
	As shown in the proof of Lemma~\ref{lem:hex}, no configuration with three consecutive positions of the same color can be universal for full-crossing rays, because it cannot be realized from a set of alternating points in convex position. 
	On the other hand, we now show that, from the point set $S=R\cup B$ described in the proof of Theorem~\ref{thm:nsquarecrossings}, only configurations with three consecutive blue positions can be realized by full-crossing rays. 
	Consider then such a point set if $n$ is even. If $n$ is odd, put $\lfloor n/2 \rfloor$ points in one of the disks and $\lceil n/2 \rceil$ points in the other one (see Figure~\ref{fig:nsquared}).
	
	\begin{figure}[!htb]
		\centering
		\includegraphics{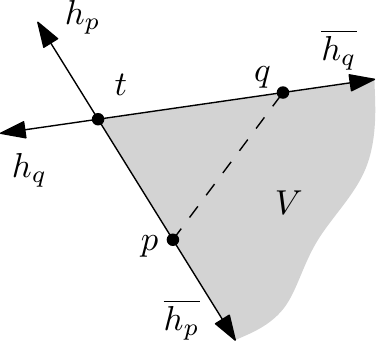}
		\caption{Observation in the proof of Theorem~\ref{prop:nouniv}. }
		\label{fig:fullcross}
	\end{figure}
	
	Consider two different points $p,q \in S$ and let $h_p$ be a ray with apex at $p$, $h_q$ be a ray with apex at $q$, and $t= h_p \cap h_q \ne \emptyset$ be their intersection point.
	Assume that the rays $h_q,h_p,\overline{h_q}$ and $\overline{h_p}$ appear clockwise in this order at infinity (see Figure~\ref{fig:fullcross}).
	If a ray $h$ from a set $H \ni h_p,h_q $ of full-crossing rays appears clockwise between $h_q$ and $h_p$, then its apex must be in the convex hull $V$ of $\overline{h_q}, \overline{h_p}$ and $t$.  
	In addition, if the apex  of $h$ does not belong to the triangle defined by $p, q$ and $t$, then $h$ must intersect the segment $pq$. 
	
	Note that if one red ray crossed $\mathcal{C}$ and another red ray crossed $\mathcal{C}'$, then these two rays do not cross. 
	Thus, at least one of the circles, say  $\mathcal{C}$, cannot be crossed by any red ray. 
	We now show that this leads to a contradiction. 
	Let $P=\{p_1,p_2,p_3,p_4,p_5\} \subset \mathcal{C}\cap B$.
	In particular, no ray in~$R$ intersects any of the segments $p_i,p_j$ for $i,j \in [5]$. 
	Let $h_i\in H$ be the ray emanating from $p_i$ for $i\in [5]$.  
	If the rays from $P$ were not interleaved by at least two red rays, the configuration would have at least three consecutive blue positions.
	As a consequence of the previous observation, there must be two points $q_1,q_2 \in R$ in the triangle defined by two points in $P$, say $p_1$ and $p_2$ and the intersection point $q$ of the rays emanating from them.
	Assume $q_1=(-i_1,0)$ and $q_2=(-i_2,0)$ for suitable $i_1,i_2 \in [n]$. 

It is now easy to see that since the distance between $p_1$ and $p_2$ is at most $2\epsilon$ (where $\epsilon$ is the the radius of $\mathcal{C}$), if $\epsilon$ is small enough,
%	It is now easy to see that since
%	the distance between $p_1$ and $p_2$ is at most $\epsilon$, if $\epsilon$ is small enough, 
then the triangle with vertices $p_1$, $p_2$ and $q$ cannot contain $q_1$ and $q_2$, a contradiction. 
\end{proof}

As stated before, any $\chi$-feasible configuration from $S$ is also feasible. 
The construction in the previous lemma implies that the converse is in general not true.
Moreover, we now show that $\gamma$ and $\mu$ differ asymptotically over some families of point sets. 

\begin{theorem}\label{prop:fullvsno}
There exist point sets $S=R \cup B$ whose number of feasible configurations is asymptotically larger than the number of $\chi$-feasible configurations.
\end{theorem}

Before proceeding to prove the previous theorem, we establish two technical lemmas. 

\begin{lemma}\label{lemm:3n}
	Let $\Sigma_n$ be the set of sequences of length $n$ on the alphabet $\{ 1,2,3,4\} $ that do not contain $ 12, 14, 132, 134, 432$ or $434$. Then, $|\Sigma_n|=\Theta (3^n)$.
\end{lemma}

\begin{proof}
	Let $h(n)=|\Sigma_n|$, and let $h_1(n),h_2(n),h_3(n)$ and $h_4(n)$ be the number of sequences of $\Sigma_n$ finishing with the symbols $1,2,3$ and $4$, respectively. Thus, $h(n)=h_1(n)+h_2(n)+h_3(n)+h_4(n)$ for any $n\ge 1$. Clearly, if $n\ge 2$, then $h_1(n)=h(n-1)=h_1(n-1)+h_2(n-1)+h_3(n-1)+h_4(n-1)$, because by removing the last symbol of a sequence of $\Sigma_n$ finishing with $1$, we obtain a sequence of the set $\Sigma_{n-1}$ and viceversa. In the same way, $h_3(n)=h(n-1)=h_1(n)$. However, given a sequence of $\Sigma_n$ finishing with $2$, the previous symbol cannot be $1$, and, if it is $3$, then the one before this $3$ can be neither $1$ nor $4$. Hence, $h_2(n)=h_4(n-1)+h_2(n-1)+h_3(n-2)+h_2(n-2)$ when $n\ge 3$. By symmetry, $h_4(n)=h_2(n)$. Therefore, the vector $(h_1(n),h_2(n),h_1(n-1),h_2(n-1))^T$ satisfies the recurrence
	$$\left( \begin{array}{lll} h_1(n)\\ h_2(n)\\ h_1(n-1)\\ h_2(n-1)\end{array}\right) =\left( \begin{array}{llll}2 & 2 & 0 & 0\cr 0 & 2 & 1 & 1\cr
	1& 0 & 0 & 0\cr 0 & 1 & 0 & 0\cr \end{array}\right)  \left( \begin{array}{lll} h_1(n-1)\cr h_2(n-1)\cr h_1(n-2)\cr h_2(n-2)\end{array}\right) $$
	for any $n\ge 3$. The eigenvalues of the $4\times 4$ matrix of the above recurrence are $3,1,0$ and $0$, so $h_1(n)$ and $h_2(n)$ (and also $h(n)$) are $\Theta (3^n)$.
\end{proof}

\begin{lemma}\label{lemm:fibo}
	Let $\Sigma'_n$ be the set of sequences of length $n$ on the alphabet $\{ 1,2\} $ that contain neither $111$ nor $222$ as subsequences. 
	Then $|\Sigma'_n|=\Theta (\Phi ^n)$, with $\Phi ={{1+\sqrt{5}}\over 2} \approx 1.618$.
\end{lemma}

\begin{proof}
	Let $f(n)=|\Sigma'_n|$, and let $f_1(n)$ and $f_2(n)$ be the number of sequences of $\Sigma'_n$ finishing with $1$ and $2$, respectively. Clearly, $f_1(n)=f_2(n-1)+f_2(n-2)$ when $n\ge 3$, because by removing the last symbol of a sequence of $\Sigma'_n$ finishing with $1$, we obtain either a sequence of the set $\Sigma'_{n-1}$ finishing with $2$, or a sequence of $\Sigma'_{n-1}$ finishing with $1$ and whose previous symbol is $2$. 
	In the same way, we have $f_2(n)=f_1(n-1)+f_1(n-2)$. Therefore, $f(n)$ satisfies the well-known Fibonacci recurrence $f(n)=f(n-1)+f(n-2)$. As a consequence, $f(n)=\Theta (\Phi ^n)$ with $\Phi = {{1+\sqrt{5}}\over 2}$.
\end{proof}

\begin{proof}{ (of Theorem~\ref{prop:fullvsno}) }
	Let $S$ be a set of $2n$ points on a semi-circle with endpoints on a horizontal line, alternating between red and blue. We first give an upper bound on the number of configurations realizable by full-crossings rays from~$S$. 
	Note that there cannot be three consecutive points of the same color in any configuration realizable by full-crossing rays from $S$. 
	Otherwise, we can choose three points of the opposite color and obtain the vertices of a convex hexagon, alternating in color. The corresponding rays are a realization with full crossings of $(r,r,r,b,b,b)$, contradicting Lemma~\ref{lem:hex}. On the other hand, by Lemma~\ref{lemm:fibo}, the number of sequences using $2n$ symbols of two colors such that no three consecutive symbols have the same color is $\Theta (\Phi ^{2n})$. Therefore, the number of configurations realizable by full-crossing rays from $S$ is at most $\Theta (\Phi ^{2n}) \approx \Theta (2.618^{n})$.
	 
	Now, let us see that, shooting vertical rays up and down from $S$, there are at least $\Omega({3^n\over n^2})$ feasible configurations. As we are shooting vertical rays, we may assume that the $2n$ points are on a line, alternating in color, denoted $p_1, q_1, p_2, q_2, \ldots p_n, q_n$ from left to right. 
	For every pair $p_i, q_i$ of points, the rays can be shot in four different ways: both rays upwards (type $1$), both rays downwards (type $2$), the red ray upwards and the blue one downwards (type $3$) and the red ray downwards and the blue one upwards (type $4$). 
	Hence, any sequence $\sigma$ of $n$ symbols from $\Lambda=\{1, 2, 3, 4\}$ produces a feasible configuration $C$. 
	Given a sequence $\sigma$, we denote by $u(\sigma)$ the sequence realized by the rays shot upwards from left to right and by $d(\sigma)$ the  sequence realized by the rays shot downwards from left to right.  
	See Figure~\ref{fig:upanddown} for an example. 
	We say that two of these sequences of $n$ symbols, $\sigma$ and $\sigma'$, are \emph{equivalent} if $u(\sigma) = u(\sigma')$ and $d(\sigma) = d(\sigma')$.
	
	\begin{figure}[h!]
		\centering
		\includegraphics[scale=1]{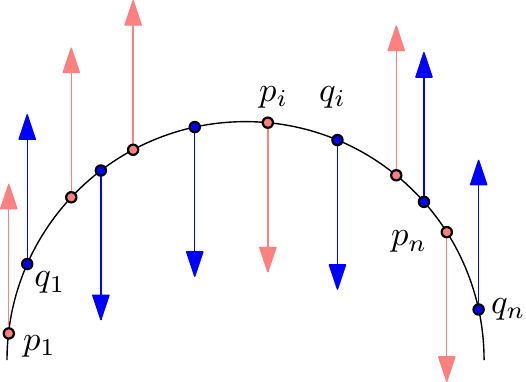}
		\caption{The feasible configuration $(rbrrrbbrbrbb)$ corresponding to sequence $\{133214\}$. The up and down configurations are $(rbrrrbb)$ and $(bbrbr)$, respectively.}
		\label{fig:upanddown}
	\end{figure}
	
	Let us consider the set $\Sigma_n$ of sequences of length $n$ on the alphabet $\Lambda$ that do not contain $12, 14, 132, 134, 432$ or $434$ as subsequences. 
	Let us prove by induction on $n$ that any two of these sequences $\sigma$ and $\sigma'$ are nonequivalent. 
	When $n=1$ the result is obvious: for sequence $\{1\}$, $u(\sigma) = (r,b)$ and $d(\sigma) = \emptyset$; for sequence $\{2\}$, $u(\sigma) = \emptyset$ and $d(\sigma) = (r,b)$; for sequence $\{3\}$, $u(\sigma) = (r)$ and $d(\sigma) = (b)$; and for sequence $\{4\}$, $u(\sigma) = (b)$ and $d(\sigma) = (r)$. Let us assume that $n>1$ and let $\sigma$ and $\sigma'$ be two of these sequences. We consider several cases depending on the first symbol of $\sigma$ and $\sigma'$.
	
	\begin{enumerate}[i)]
		\item Let $\sigma$ and $\sigma'$ begin with $1$. By removing the first $1$ in both sequences we obtain sequences $\overline{\sigma}$ and $\overline{\sigma}'$ of length $n-1$, and the result follows immediately by induction on them. 
		The same argument holds when $\sigma$ and $\sigma'$ both begin with $2, 3$ or $4$.
		\item Let $\sigma$ begin with $1$ and $\sigma'$ begin with $2$. 
		As $\sigma'$ begins with 2, then the two first elements appearing in $d(\sigma')$ are $r$ and $b$. 
		As $\sigma$ begins with $1$ and the subsequences $12$ and $14$ are forbidden, then either $\sigma$ contains only symbols 1, or it begins with a sequence of symbols $1$ followed by a symbol $3$. 
		This implies that $d(\sigma)$ is empty or it begins with a symbol $b$, so $\sigma$ and $\sigma'$ are nonequivalent.
		\item Let $\sigma$ begin with $1$ and $\sigma'$ begin with $3$. 
		Assuming that $\sigma$ and $\sigma'$ are equivalent, a contradiction is reached. 
		Since $\sigma$ begins with a symbol $1$, then $u(\sigma)$ begins with $r$ and $b$. 
		As~$\sigma'$ begins with a symbol $3$, $u(\sigma')$ begins with $r$. 
		The only way of having $b$ in the second position of $u(\sigma')$ is that~$\sigma'$ begins with $3$ and then a symbol $4$ appears after several symbols $2$ (if any). 
		This implies that $d(\sigma')$ begins with $b$ and $r$. 
		As $\sigma$ begins with $1$, the only way to have $b$ as the first element of $d(\sigma)$ is that, after maybe several symbols $1$, a symbol $3$ appears, because the subsequences $12$ and $14$ are forbidden. 
		So, the beginning of $\sigma$ is $\{11 \ldots  1 3\}$. 
		However, since sequences $132$ and $134$ are forbidden, it is impossible that the second element of $d(\sigma)$ is $r$, contradicting the assumption that $\sigma$ and $\sigma'$ are equivalent. 
		By symmetry, $\sigma$ and~$\sigma'$ are also nonequivalent when $\sigma$ begins with $2$ and $\sigma'$ with $4$.
		\item Let $\sigma$ begin with $1$ and $\sigma'$ begin with $4$. Clearly, both sequences are nonequivalent because the first element of $u(\sigma)$ is $r$ and the first element of $u(\sigma')$ is $b$. A similar reasoning applies when $\sigma$ begins with $2$ and $\sigma'$ with $3$, and when $\sigma$ begins with $3$ and $\sigma'$ with $4$.
	\end{enumerate}
	
	By Lemma~\ref{lemm:3n}, the number of nonequivalent sequences such that subsequences $12, 14, 132,$ $134, 432$ and $434$ are forbidden is $\Omega (3^n)$. We next argue that two nonequivalent sequences can produce the same feasible configuration, but this can happen only a quadratic number of times. 
	By choosing two indices $i<j$ such that $j-i$ is even, a feasible configuration $C=\{c_1, \ldots , c_{2n}\}$ can be divided into two subsequences $C_1=\{c_i, \ldots c_{j-1}\}$ and $C_2 = \{c_{j}, \ldots , c_{2n}, c_1, \ldots , c_{i-1}\}$. 
	By defining $C_1$ (or $C_2$) as $u(\sigma)$ and the inverse of $C_2$ (or $C_1$) as $d(\sigma)$, we obtain the up and down configurations associated to a sequence of $n$ symbols. Since as a feasible configuration can be divided into at most a quadratic number of different ways, at most a quadratic number of sequences can produce the same feasible configuration. 
	As a consequence, the number of feasible configurations is at least $\Omega ({3^n / n^2})$.	
\end{proof}

\section{Final remarks and open questions}

The decision algorithm for points in general position presented in Section~\ref{subsec:deciding} can be improved if the point set is in convex position. 
In this case, it is not hard to see that part $(i)$ of Lemma~\ref{lem:decomp} always holds. 
Therefore, it is not necessary to compute the table $T_\Sigma$ and, hence, the complexity of the algorithm is $O(n^{10})$ time and $O(n^7)$ space. 
However, it remains open whether this runtime, and even the one for the general case, can be improved.

The main open question left to future work is to give a lower bound for the number of configurations realizable by full-crossing rays from any point set. We conjecture that, for point sets in strong general position, this number is always $\Omega(2^{\sqrt{n}})$.
It would be also interesting to prove the tightness of some of the bounds in this paper or improve them otherwise.

\paragraph{Acknowledgments.} This work started at the \textsl{2nd International Workshop on Discrete and Computational Geometry}, held in Barcelona in July 1-5, 2013, and organized by Ferran Hurtado and Rodrigo Silveira.

R.~Fabila-Monroy is partially supported by CONACyT of Mexico grant 253261.
F.~Hurtado, R.~Silveira, J.~Tejel and A.~Garc\'{\i}a 
are partially supported by projects MINECO MTM2012-30951, MTM2015-63791-R, ESF EUROCORES programme EuroGIGA, CRP ComPoSe: MICINN Project EUI-EURC-2011-4306.
F.~Hurtado and R.~Silveira are partially supported by Gen. Cat. DGR2009SGR1040 and DGR2014SGR46.
A.~Garc\'{\i}a and J.~Tejel are partially supported by Gob. Arag. E58 (ESF).
R.~Jaume is supported by ``obra social La Caixa'' and the DAAD. 
P.~P\'{e}rez-Lantero is partially supported by projects CONICYT FONDECYT/Regular 1160543 (Chile), 
and Millennium Nucleus Information and Coordination in Networks ICM/FIC RC130003 (Chile).
M.~Saumell is supported by projects NEXLIZ - CZ.1.07/2.3.00/30.0038, which is co-financed by the European Social Fund and the state budget of the Czech Republic, ESF EuroGIGA project ComPoSe as F.R.S.-FNRS - EUROGIGA NR 13604, and project LO1506 of the Czech Ministry of Education, Youth and Sports.
R.~Silveira was partially funded by MINECO through the Ram{\'o}n y Cajal program.
J.~Urrutia is supported by CONACyT of Mexico grant 178379, and  Programa de Apoyo a Proyectos de Investigaci\'on e Innovaci\'on Tecnol\'ogica, UNAM, grant number IN102117.
We thank Javier Cano, Clemens Huemer, Alberto M\'{a}rquez, and Vera Sacrist\'{a}n for helpful discussions. This project has received funding from the European Union’s Horizon 2020 research and innovation programme under the Marie Skłodowska-Curie grant agreement No 734922.

\small
\bibliographystyle{abbrv}

\begin{thebibliography}{99}


\bibitem{AHHKV07}
O. Aichholzer, T. Hackl, C. Huemer, F. Hurtado, H. Krasser, and B. Vogtenhuber. On the number of plane geometric graphs. \emph{Graphs and Combinatorics} 23(1) (2007), 67--84.

\bibitem{AU907}
J. Akiyama and J. Urrutia. Simple alternating path problem.
\emph{Discrete Mathematics} 84(1) (1990), 101--103.

\bibitem{CCL13}
S. Cabello, J. Cardinal, and S. Langerman. The Clique Problem in Ray Intersection Graphs.
\emph{Discrete \& Computational Geometry} 50(3) (2013), 771--783.

\bibitem{DSST13}
A. Dumitrescu, A. Schulz, A. Sheffer, and C. D. T\'{o}th. Bounds on the maximum multiplicity of some common geometric graphs. \emph{SIAM Journal on Discrete Mathematics} 27(2) (2013), 802--826.

\bibitem{ES35}
 P. Erd\"os, G. Szekeres. A combinatorial problem in geometry. \emph{Compositio Mathematica} 2 (1935),  463--470.

\bibitem{FMM13}
S. Felsner, G. B. Mertzios, and I. Musta. On the Recognition of Four-Directional Orthogonal Ray Graphs.
\emph{Proc. 38th International Symposium on
Mathematical Foundations of Computer Science}, volume 8087 of Lecture Notes in Computer Science, pp. 373--384. Springer, 2013.

\bibitem{GHTU07} A. Garc\'{\i}a, F. Hurtado, J. Tejel, and J. Urrutia. Configurations of non-crossing rays and related problems. \emph{Discrete \& Computational Geometry} 55(3) (2016), 522--549.

\bibitem{GOR04} J. E. Goodman and J. O'Rourke, editors. Handbook of discrete and computational
geometry. CRC Press, Inc., Boca Raton, FL, USA, second edition, 2004.

\bibitem{HSSTW11}
    M. Hoffmann, M. Sharir, A. Sheffer, C. D. T\'{o}th, and E. Welzl.  Counting Plane Graphs: Flippability and its Applications. \emph{Proc. 12th International Symposium on Algorithms and Data Structures}, volume 6844 of Lecture Notes in Computer Science, pp. 524--535. Springer, 2011.

\bibitem{Juk11}
S. Jukna. \emph{Extremal Combinatorics With Applications in Computer Science}, Springer-Verlag (Second Edition), Series \emph{Texts in Theoretical Computer Science} XXIV, page 71 exercise 4.12.

\bibitem{KK03}
A. Kaneko, and M. Kano.
Discrete geometry on red and blue points in the plane -- a survey.
Discrete and Computational Geometry, The Goodman-Pollack Festschrift.
Springer, Algorithms and Combinatorics series, Volume 25, 2003, pp. 551--570.

\bibitem{KYZ13} D. Kirkpatrick, B. Yang and S. Zilles. On the barrier-resilience of arrangements of ray-sensors.\emph{ Proc. of the XV Spanish Meeting on Computational
Geometry}, Seville, Spain, June 2013, pp. 35--38.

\bibitem{M1872}
C. Moreau. Sur les permutations circulaires distinctes. \emph{Nouvelles annales de math\'{e}matiques, journal des candidats aux \'{e}coles polytechnique et normale}, S\'{e}r. 2, tom.~11 (1872), 309--314 

\bibitem{SS09}
    M. Sharir and A. Sheffer. Counting triangulations of planar point sets. \emph{The Electronic Journal of Combinatorics}, 18(1) (2011).

\bibitem{SSW13}
    M. Sharir, A. Sheffer, and E. Welzl. Counting Plane Graphs: Perfect Matchings, Spanning Cycles, and Kasteleyn's Technique. \emph{Journal of Combinatorial Theory, Series A}, 120 (2013), 777--794.

\bibitem{SW06}
M. Sharir and E. Welzl. On the number of crossing-free
matchings, cycles, and partitions. \emph{SIAM Journal on Computing} 36(3) (2006), 695--720.

\bibitem{MTU10}
A. Shrestha, Sa. Tayu, and S. Ueno. On orthogonal ray graphs. \emph{Discrete Applied Mathematics} 158(15) (2010), 1650--1659.

%\bibitem{sloane}
%N. J. A. Sloane. \emph{The On-Line Encyclopedia of Integer Sequences.}
%\url{ http://oeis.org}. Sequence A001519.

\bibitem{vanLint}
J.H. van Lint and R.M. Wilson. A course in Combinatorics. Cambridge University Press. Cambridge, Great Britain, 1992.

\end{thebibliography}

\end{document}